\documentclass[12pt,a4paper,english]{article}
\usepackage{geometry}
\usepackage[T1]{fontenc}
\usepackage[utf8]{inputenc}
\usepackage[english]{babel}
\usepackage{soul}

\usepackage{lipsum}
\usepackage{url}
\usepackage[hidelinks]{hyperref}

\usepackage[inline]{enumitem}

\usepackage{graphicx}
\usepackage{booktabs}
\usepackage{tikz}
\usetikzlibrary{calc,patterns,angles,quotes,arrows.meta}
\usepackage{tikz-cd}
\usepackage[labelfont=bf]{caption}
\usepackage{amsmath}
\usepackage{mathrsfs}
\usepackage{amssymb}
\usepackage{wasysym}
\usepackage{amsthm}
\usepackage{mathtools}
\usepackage{arydshln}
\usepackage{algorithm}
\usepackage{bigints}
\makeatletter
\renewcommand{\ALG@name}{Pseudo-algorithm}
\makeatother
\usepackage{algpseudocode}

\newtheorem{theorem}{Theorem}
\newtheorem{proposition}{Proposition}
\theoremstyle{definition}
\newtheorem{remark}{Remark}
\newtheorem*{comment*}{Comment}

 \def\norm#1{\left \vert #1 \right \vert}
 \def\Norm#1{\left \Vert #1\right \Vert}
\def\cvarGamma{c_5}
\def\cdelta{\gamma_1}
\def\cderdelta{\gamma_2}
\def\cU{c_6}

 \title{
A Kustaanheimo-Stiefel regularization of the elliptic restricted three-body problem and the detection of close encounters with fast Lyapunov indicators
 }
\author{Mattia Rossi\textsuperscript{1,2} and Massimiliano Guzzo\textsuperscript{1}\\
	\small{\textsuperscript{1}Università degli Studi di Padova}\\
	\small{Dipartimento di Matematica ``Tullio Levi-Civita''}\\
	\small{Via Trieste, 63 - 35121 Padova, Italy}\\
	\small{mrossi@math.unipd.it, guzzo@math.unipd.it	
		\vspace{1mm}}\\	
	\small{\textsuperscript{2}Università degli Studi di Genova}\\
	\small{MIDA -- Dipartimento di Matematica}\\
	\small{Via Dodecaneso, 35 - 16146 Genova, Italy}\\
	\small{rossi.ma@dima.unige.it}}
\date{\today}

\begin{document}

\maketitle

\begin{abstract}
We present the Kustaanheimo-Stiefel (KS) regularization of the elliptic restricted three-body problem (ER3BP) at the secondary body $P_2$, and discuss its use to study a category of transits through its Hill's sphere (fast close encounters). Starting 
from the Hamiltonian representation of the problem using the synodic 
rotating-pulsating reference frame and the  true anomaly  of $P_2$ as independent variable, we perform the regularization at the secondary body analogous to the circular case by applying the classical KS transformation and the iso-energetic reduction in an extended 10-dimensional phase-space. Using such regularized Hamiltonian we recover a definition of fast close encounters in the ER3BP for small values of the mass parameter $\mu$ (while we do not
require a smallness condition on the eccentricity of the primaries), and we show that for these encounters the solutions of the variational equations are characterized by an exponential growth during the fast transits through the Hill's sphere. Thus, for small $\mu$, we
justify the effectiveness of the regularized fast Lyapunov indicators (RFLIs) to detect orbits with multiple fast close encounters. Finally, we provide numerical demonstrations and show the benefits of the regularization in terms of the computational cost.\\

\end{abstract}

\section{Introduction}
\label{sec:intro}
The regularization of the gravitational singularities, appeared at the beginning of the XXth century, has become in the last decades an extremely useful technique to deal with the numerical integration of the N-body problems. Particularly two kinds of regularization techniques are widely known: a geometric one, which basically aims at modifying the equations of motion such that they are defined and regular even on the singularities, and a solution-based one, whose goal consists in an analytic continuation of the original solution through the singular point. In this paper we focus precisely on a celebrated example of the former category, i.e., the Kustaanheimo-Stiefel regularization for a special case of utmost importance, represented by the restricted three-body problem, originally in its circular (CR3BP) and then elliptic (ER3BP) variant.

In a seminal paper \cite{LC1906} Levi-Civita performed a local\footnote{In case of multiple singularities the term refers to the deletion of only one of them at a time, as opposed to a global method, mainly due to Birkhoff \cite{Birkhoff:restricted}.} regularization of the planar CR3BP, which relies on the conservation of the so-called Jacobi integral, through the introduction of canonical transformations and a time reparametrization that nowadays are known, after his name, as Levi-Civita (LC) regularization. The issue for the spatial CR3BP was solved by Kustaanheimo and Stiefel in the mid-1990s \cite{K64, KS65}. The latter procedure is more complicated than the LC one since it exploits a projection map from a space of four redundant variables to the three-dimensional Cartesian space. Both LC and KS regularizations are {\it iso-energetic}, since they exploit the existence of a global first integral, the so-called Jacobi constant. 
There exist in the literature many uses of the KS transformation to regularize binary collisions in the general 3-body and
 N-body problems, as well as perturbations of the Kepler problem whose definitions include
 restricted problems more general than the CR3BP and ER3BP (see \cite{StiefelScheile1971} and, for example,  \cite{aarseth-zare,heggie,shefer1990,Falcolini,aarseth2003,Waldvogel:2006,langnerbreiter2015,breiter2017} and references therein). For the ER3BP (and for more general restricted problems) a 10-dimensional phase-space is required. In fact, in addition to the 8-dimensional phase-space of
 the KS variables and an additional variable corresponding to the physical time,
 another variable corresponding to the energy of the
 system (variable in time for the ER3BP) is required. While
 these approaches include the regularizations of the ER3BP, the geometric properties of
 this system provide specific representations which are more adapted, for example, to the computation of ejection-collision orbits or the dynamics 
 at the Lagrangian points (see, e.g., \cite{SZ:thorb,PG21}) and, as we consider in
 this paper, the study of fast close encounters with the secondary body $P_2$.  

We remark that specific approaches to the regularization of the ER3BP
have been developed also using regularizations different from the KS one. Formulations in this regard have been derived in \cite{SZ:thorb,Szebehely:regER3BP} and applied in \cite{Broucke,Pinyol} for the planar setting; in \cite{Waldvogel:regER3BP,Waldvogel:NASA,Arenstorf:regER3BP,Llibrepinol} for the spatial one.\\

A convenient formulation of the ER3BP uses a synodic rotating-pulsating reference frame and the  true anomaly  of the secondary body as independent variable \cite{SZ:thorb}. Consider an ER3BP defined by the motion of a body $P$ of negligible mass in the gravitation field of two massive bodies $P_1$ (the primary) and $P_2$ (the secondary) performing an elliptic Keplerian motion of eccentricity $\varepsilon\in (0,1)$. As usual, the simplifying assumptions on the units correspond to setting $m_1=1-\mu$, $m_2=\mu$ for $\mu\in(0,1/2)$ as the masses of $P_1,P_2$ respectively, while $a=1$ and $T=2\pi$ are the semi-major axis and the period of the elliptic motion.
By denoting with $(x,y,z)$, $(p_1,p_2,p_3)$ the coordinates of $P$ and their conjugate momenta, and with $f$ the true anomaly of the elliptic motion of $P_2$ which is used as independent variable, the Hamiltonian reads:
\begin{multline}
\label{eqn:ER3BP3Dhamfintro}
\mathcal{H}(x,y,z,p_1,p_2,p_3,f)=\frac12(p_1^2+p_2^2+p_3^2)+p_1y-xp_2\\
-\frac{1}{1+\varepsilon\cos f}\left(\frac{1-\mu}{d_1}+\frac{\mu}{d_2}-\frac{1}{2}(x^2+y^2+z^2)\varepsilon\cos f\right)\;,
\end{multline}
where $d_1=\Vert P-P_1\Vert$, $d_2=\Vert P-P_2\Vert$. In this paper we first represent the KS regularization at the secondary body of the Hamiltonian  (\ref{eqn:ER3BP3Dhamfintro}) by applying the classical KS transformation and the iso-energetic reduction in the extended 10-dimensional phase-space. First, we provide indeed a simple and self-consistent proof
on the projection of the solutions of the regularized Hamiltonian to the original solutions, by adapting to the elliptic case a derivation of the KS Hamiltonian of the CR3BP given in \cite{GuzzoCardin:intCR3BP}. Then, we also assess the effectiveness of the regularization for numerical integrations, and the advantage in terms of numerical performances, in a fictitious simple scenario which is nevertheless representative (for the choice of the initial conditions) of realistic close encounters in the Solar System, such as the non-coplanar close encounters with Jupiter, in the Sun-Jupiter ER3BP. Finally, we use the regularized Hamiltonian to extend the definition of fast close encounters to the ER3BP, and justify the effectiveness of the RFLIs to detect orbits with multiple of such fast close encounters.\\

We call `close encounter' a transit of a solution $(x(f),y(f),z(f))$ of the ER3BP through the Hill's sphere of $P_2$ occurring in any finite interval $[f_1,f_2]$ of the true
anomaly $f$ (collision solutions are not included). We consider values of the mass parameter $\mu\leq 1/10$ and we define the Hill's sphere of $P_2$ by:
$$
B(\mu^{1\over 3})= \{(x,y,z)\in {\Bbb R}^3\colon d_2< \mu^{1\over 3}\}\;.
$$
Notice that this definition of Hill's sphere is different from the
conventional one by a numerical factor:
the radius $\mu^{1\over 3}$ corresponds to $3^{1\over 3}r_h$, where $r_h$ is the usual Hill's radius. A close encounter occurring for $f\in [f_1,f_2]$ 
satisfies $d_2(f_1)=d_2(f_2)=\mu^{1\over 3}$ as well as $0< d_2(f) < \mu^{1\over 3}$ for all $f\in (f_1,f_2)$. 

We consider the category, critical for the numerical integrations, of fast close encounters, generalizing the conditions of fast close encounters
which are given for the CR3BP. 
 We recall that  fast close encounters are
frequently observed for celestial bodies in the Solar System (see for example \cite{GL15,GL17}, where the dynamics of comet 67P Churyumov-Gerasimenko, target of the recent Rosetta mission, is discussed in detail), and they are important also for the study of the risk of impact of asteroids on the Earth, as well as for the technique of gravity assist to change the energy of a spacecraft.

Fast close encounters are easily indentified in the CR3BP by the encounters occurring for values of the Jacobi constant ${\cal C}$ satisfying:
\begin{equation}
\gamma \coloneqq {3- 4 \mu +\mu^2-{\cal C}\over 2}>0\;,
\label{JacobiC}
\end{equation}
with the exclusion of a neighborhood of $\gamma=0$. The condition
(\ref{JacobiC}) (or similar ones), which appeared in several studies of close encounters with Levi-Civita regularization (see, e.g., \cite{henrard,FNS,GL2013,GKZ,Guzzo}) as well as in the heuristic approach known as \"Opik's theory (see \cite{opik}, revised in \cite{valsecchi2002}),  
is understood by representing the Hamiltonian of the planar CR3BP, Levi-Civita regularized at $P_2$, in the form:
\begin{equation}
\label{eqn:hamLCregintro6}
\mathcal{K}(u,U)=\frac{1}{8}\Vert U-b(u)\Vert^2
-\Vert {u}\Vert^2\left ({3-4 \mu+\mu^2 -{\cal C}\over 2}
\right)-\mu  +\mathcal{R}_6(u)\;,
\end{equation}
where $(u_1,u_2)$ are the Levi-Civita coordinates and $U=(U_1,U_2)$ are the momenta conjugate to $u=(u_1,u_2)$ defined as in \cite{LC1906}; $b(u)=(b_1(u),b_2(u))$ is cubic in the $u$; $\mathcal{R}_6(u)$ is regular at $u=0$ with Taylor expansion which begins with order 6. If $\gamma>0$, the coefficient of $\Norm{u}^2$ in the Levi-Civita Hamiltonian \eqref{eqn:hamLCregintro6} is strictly negative, allowing to use methods of hyperbolic dynamics to study the fast close encounters of the
CR3BP, as it was done in \cite{henrard,FNS,GKZ} with analytic methods, and in
\cite{GL2013} to study the effectiveness of RFLIs to detect orbits with multiple close enclounters in the planar CR3BP.
We remark that, by considering the higher order corrections to the quadratic approximation of the Hamiltonian, a neighborhood of the limit case $\gamma=0$ should be avoided. Alternative definitions of fast close encounters,
as in \"Opik's theory, refer to the hyperbolic approximations
of the Cartesian solutions which are obtained by considering the Keplerian motion defined by the secondary body $P_2$. 

Regularized fast Lyapunov indicators have been introduced in \cite{CLSF,LGF11,GL2013} (see also \cite{GL23}, and refeences therein) and have been used in \cite{GL2013,LG16,GL18}
to detect the several kinds of close encounters with the secondary body $P_2$. 
The ability of FLIs to detect fast close encounters of the planar CR3BP \cite{REReleg,RECharac}, i.e., for $\gamma>0$, has been related to the exponential growth of tangent vectors for orbits transiting suitably
fast in the Hill's sphere of the planet \cite{GL2013}. In fact, the Jacobian matrix ${\cal X}(u,U)$ of the Hamiltonian vector field of (\ref{eqn:hamLCregintro6}) computed in the limit:
$$
{\cal X}_0= \lim_{\Norm{u} \to 0, \Norm{\text{d}u/\text{d}s}\to\sqrt{\mu/2}}{\cal X}(u,U)\;,
$$
where $s=s(t)$ is the Levi-Civita time reparametrization, has eigenvalues $\pm \sqrt{\gamma/2}$, and therefore is hyperbolic
when $\gamma>0$.

To develop this idea for the full ER3BP we consider the Hamiltonian ${\cal K}(u,\phi,U,\Phi)$ of the problem regularized at $P_2$, as it will be obtained in Section \ref{sec:thm}: $u=(u_1,u_2,u_3,u_4)$ denote the KS variables;
$U=(U_1,U_2,U_3,U_4)$ a set of conjugate momenta. The  additional conjugate variables $\phi,\Phi$ are needed in the regularization of the elliptic problem: by denoting with
$s$ the independent variable (the proper time) of the Hamilton equations of the regularized Hamiltonian, $\phi(s)$ is the true anomaly of $P_2$ for the value $s$ of the proper time (see Section \ref{sec:thm} for all the details). We obtain for ${\cal K}(u,\phi,U,\Phi)$ a representation similar
to \eqref{eqn:hamLCregintro6}:
\begin{equation}
\mathcal{K}(u,\phi,U,\Phi)=\frac{1}{8}\Vert U-b(u)\Vert^2
-\Vert {u}\Vert^2\left (-\Phi + {3-4 \mu +\mu^2\over 2 (1+\varepsilon \cos\phi)}
  \right )-{\mu \over 1+\varepsilon \cos\phi} +\mathcal{R}_6(u,\phi)\;,
\label{eqn:khamregintro6}
\end{equation}
where $b(u)=\mathcal{O}(\Norm{u}^3)\in\mathbb{R}^4$; $\mathcal{R}_6(u,\phi)$ is regular at $u=0$, and its Taylor
expansion in the vector variable $u$ begins with order 6. There is however a fundamental difference
with respect to the circular case, since the coefficient:
\begin{equation}
  \Gamma(\phi,\Phi) =-\Phi +{3 -4 \mu +\mu^2\over 2(1+\varepsilon \cos\phi)}
    \label{Gamma-ell}
\end{equation}
depends on the variables $\Phi,\phi$, and therefore is not constant
along the solutions $(u(s),\phi(s),\allowbreak U(s),\Phi(s))$ of the Hamilton equations
of ${\cal K}$. Moreover, the derivative of $\Gamma_s \coloneqq\Gamma(\phi(s),\Phi(s))$ with respect to the proper time $s$ is proportional to $\varepsilon$, but does not vanish for $\mu\to 0$. Using the
representation (\ref{eqn:khamregintro6}),
we notice that the variation of $\Gamma_s$ satisfies:
$$
{\text{d}\over \text{d}s}\Gamma_s =-\varepsilon \mu {\sin\phi \over (1+\varepsilon \cos\phi)^2}-{\varepsilon \sin\phi \over (1+\varepsilon \cos\phi)^2} {\cal O}(\Norm{u}^6)\;,
$$
so that, during the transit through the Hill's sphere, where
$\Norm{u}$ is smaller than order $\mu^{1\over 6}$ (see Section \ref{sec:thm}), we have:
$$
\norm{\frac{\text{d}\Gamma_s}{\text{d}s}} = {\cal O}(\varepsilon \mu)\;. 
$$
The stability of $\Gamma_s$ up to times of the
order of $1/(\varepsilon\mu)$ provides the opportunity to assess the hyperbolicity
of fast close encounters for the ER3BP, for small values of $\mu$. 
In fact, to establish the hyperbolic character of
the close encounter in the ER3BP from the representation (\ref{eqn:khamregintro6}) we need to establish the stability of the coefficient $\Gamma_s$
for two reasons. First, we notice that the variational matrix ${\cal X}$ of the Hamiltonian vector field of ${\cal K}$ has the representation:
\begin{equation}
		\label{eqn:matX}
{\cal X}(u,U,\phi)= \left(
\begin{array}{cc}
  {\partial^2 {\cal K}\over \partial u \partial U} &  {\partial^2 {\cal K}\over \partial U\partial U} \\
   -{\partial^2 {\cal K}\over \partial u \partial u} &  -{\partial^2 {\cal K}\over \partial U \partial u} 
 \end{array}
\right)
= {\cal X}_0 +{\cal O}(\Norm{u}^2) + {\cal O}(\Norm{U}){\cal O}(\Norm{u})\;,
\end{equation}
with
\begin{equation}
	\label{eqn:matX0}
{\cal X}_0 = \left(
\begin{array}{cccccccc}
 0 & 0 & 0 & 0 & 1/4 & 0 & 0 & 0 \\
 0 & 0 & 0 & 0 & 0 & 1/4 & 0 & 0 \\
 0 & 0 & 0 & 0 & 0 & 0 & 1/4 & 0 \\
 0 & 0 & 0 & 0 & 0 & 0 & 0 & 1/4 \\
 2 \Gamma_s & 0 & 0 & 0 & 0 & 0 & 0 & 0 \\
 0 & 2 \Gamma_s & 0 & 0 & 0 & 0 & 0 & 0 \\
 0 & 0 & 2 \Gamma_s & 0 & 0 & 0 & 0 & 0 \\
 0 & 0 & 0 & 2 \Gamma_s & 0 & 0 & 0 & 0 \\
\end{array}
\right)\;.
\end{equation}
Therefore, on one hand we need that motions entering the Hill's sphere with $\Gamma_0 >0$, maintain a value of $\Gamma_s > 0$  during the transit
through the Hill's sphere,  so that the matrix ${\cal X}_0$ is hyperbolic.
On the other hand, we need to provide an upper bound to the elements
of the matrix ${\cal X}-{\cal X}_0$ during the transit in the Hill's sphere,
to ensure the hyperbolicity of ${\cal X}$ for suitably small
values of $\mu$. In particular, a sufficient upper bound
to $\norm{U_j(s)}$ is obtained if, for example, during the transit
we have $\Gamma_s \leq (3/2) \Gamma_0$ (see Section \ref{sec:fastcloseencounters} for all the
details). In Section  \ref{sec:fastcloseencounters} we prove that, if $\mu$ satisfies:
\begin{equation}
  \mu < c (1-\varepsilon)^6 \Gamma_0^{3\over 2}\;,
\end{equation}
where $c>0$ is a suitable constant independent of $\mu,\varepsilon,\Gamma_0$, 
then during the transit in the Hill's sphere we have 
$\Gamma_s \in [\Gamma_0/2,(3/2) \Gamma_0]$, so that the matrix ${\cal X}_0$ is hyperbolic; 
an additional smallness condition on $\mu$  grants that also the matrix ${\cal X}$ is hyperbolic. Therefore, we justify the effectiveness of the RFLIs to detect orbits with multiple fast close encounters also for the ER3BP when the parameter $\mu$ is small. We also provide numerical demonstrations of the detection of close encounters with regularized Lyapunov
indicators.\\ 

The paper is structured as follows. In Section \ref{sec:thm} we present a step-by-step construction of the regularization with the final rigorous statement on the projection of the solutions and related proof; Section \ref{sec:fastcloseencounters} is dedicated to the discussion of fast close encounters
  of the ER3BP. Section \ref{sec:numericexamples} is dedicated to numerical demonstrations:  
  Subsection \ref{sec:numeric}, after a short description of the considered scenarios, deals with numerical explorations in a neighborhood of $P_2$ and outlines quantitatively the gain as regards the computational effort;
  Subsection \ref{sec:RFLI} reports examples of detection of fast close encounters with RFLIs in the ER3BP. The details about the transformations which are needed to implement numerically the KS regularization are given in the appendix.

\section{KS regularization of the ER3BP in the synodic reference frame}
\label{sec:thm}

As described in Section \ref{sec:intro}, consider the Hamiltonian \eqref{eqn:ER3BP3Dhamfintro} of the ER3BP, which is conveniently expressed in a synodic rotating-pulsating Cartesian frame \cite{SZ:thorb} where the bodies $P_1,P_2$ have coordinates $(-\mu,0,0)$, $(1-\mu,0,0)$ respectively.\\
\indent Let us now introduce a local regularization on the secondary body\footnote{The local regularization at the primary body $P_1$ could be introduced following the same scheme. We here focus on the regularization at $P_2$ which is particularly relevant for applications to the motion of asteroids, comets and space-flight dynamics.} $P_2$. Following \cite{K64,KS65} we introduce the KS space map as a projection from a space of redundant variables $u_1,u_2,u_3,u_4$ to a space of Cartesian variables $q_1,q_2,q_3$:
\begin{align}
\label{eqn:KSproj}
\begin{split}
\pi\colon\mathbb{R}^4&\longrightarrow\mathbb{R}^3\\
u=(u_1,u_2,u_3,u_4)&\longmapsto\pi(u)=(q_1,q_2,q_3)=q\;,
\end{split}
\end{align}
where
\begin{equation}
\label{eqn:KSprojrel}
(q_1,q_2,q_3,0)=A(u)u\;
\end{equation}
are related to $(x,y,z)$ by
\begin{equation}
\label{eqn:trans}
(x-1+\mu,y,z)=q\;
\end{equation}
and
\begin{equation}
\label{eqn:KSA}
A(u)=
\begin{pmatrix}
u_1 & -u_2 & -u_3 & u_4\\
u_2 & u_1 & -u_4 & -u_3\\
u_3 & u_4 & u_1 & u_2\\
u_4 & -u_3 & u_2 & -u_1
\end{pmatrix}
\end{equation}
is a matrix that plays a central role in the KS regularization. In particular, $A(u)$ fulfills the two properties: it is a linear homogeneous function of $u_1,u_2,u_3,u_4$ and satisfies:
\begin{equation}
	\label{eqn:Aprop}
	A(u)A^T(u)=A^T(u)A(u)=\Vert u\Vert^2\mathbb{I}\;,
\end{equation}
where $\mathbb{I}$ is the $4$-by-$4$ identity matrix; hence $\Vert u\Vert^2=d_2$.\\

In this article we exploit directly such transformation in the elliptic framework by adapting the Hamiltonian derivation of the KS regularization developed in \cite{GuzzoCardin:intCR3BP} for the  CR3BP. We prove that a KS regularization with respect to the secondary body $P_2$ of the ER3BP is represented by the Hamiltonian:
\begin{multline}
\label{eqn:hamregintro}
\mathcal{K}(u,\phi,U,\Phi)=\frac{1}{8}\Vert U-b(u)\Vert^2-\frac{1}{1+\varepsilon\cos\phi}\bigg[(1-\mu)\Vert u\Vert^2\bigg(\frac{1}{\Vert\pi(u)+(1,0,0)\Vert}+\pi_1(u)\bigg)\\
+\mu+\frac 12\Vert u\Vert^2\left(\pi_1^2(u)+\pi_2^2(u)-\pi_3^2(u)\varepsilon\cos \phi\right)+\frac{(1-\mu)^2}{2}\Vert u\Vert^2\bigg]+\Phi\Vert u\Vert^2\;,
\end{multline}
where $u$, $U=(U_1,U_2,U_3,U_4)$ are the KS variables and their conjugate momenta and $\Phi$ is an action conjugate to $\phi$ introduced to make autonomous the problem. The vector $b(u)$ is defined by:
\begin{equation}
\label{eqn:b}
b(u)=2A^T(u)\Lambda A(u)u\;,\quad\Lambda=
\begin{pmatrix}
0 & -1 & 0 & 0\\
1 & 0 & 0 & 0\\
0 & 0 & 0 & 0\\
0 & 0 & 0 & 0
\end{pmatrix} \;.
\end{equation}
Specifically we show that, by adopting a fictitious time $s$ as new independent variable, the solutions $(u(s),\phi(s),U(s),\Phi(s))$ of Hamilton equations related to $\mathcal{K}(u,\phi,U,\Phi)$:
\begin{equation*}
\frac{\text{d}u}{\text{d}s}=\frac{\partial\mathcal{K}}{\partial U}\;,\quad\quad
\frac{\text{d}\phi}{\text{d}s}=\frac{\partial\mathcal{K}}{\partial\Phi}\;,\quad\quad
\frac{\text{d}U}{\text{d}s}=-\frac{\partial\mathcal{K}}{\partial u}\;,\quad\quad
\frac{\text{d}\Phi}{\text{d}s}=-\frac{\partial\mathcal{K}}{\partial\phi}\;,
\end{equation*}
that, for $s=0$, satisfy:
\begin{enumerate}[label=(\roman*)]
	\item $u(0)\neq0$,
	\item $l(u(0),U(0))=0$, with $l(u,U)=u_4U_1-u_3U_2+u_2U_3-u_1U_4$,
	\item $\mathcal{K}(u(0),\phi(0),U(0),\Phi(0))=0$, $\phi(0)=f_0$
\end{enumerate}
project (via $\pi$, the translation $x\mapsto x+1-\mu$ and $\text{d}f/\text{d}s=\Vert u\Vert^2=d_2$), locally to $s=0$, onto solutions $(x(f),y(f),z(f),p_1(f),p_2(f),p_3(f))$ of Hamilton equations
from \eqref{eqn:ER3BP3Dhamfintro}. \\

\subsection{Lagrangian formulation in the rotating-pulsating frame}
\label{subsec:LagrRP}
Let $L(x,y,z,x^{\prime},y^{\prime},z^{\prime},f)$ be the Lagrangian of the spatial ER3BP in the rotating-pulsating frame with explicit dependence on the true anomaly\footnote{The change of time from $t$ to $f$ is a classic simplification \cite{scheibner1866satz,SZ:thorb}, where $f$ is thought taking values in the covering $\mathbb{R}$ of $\mathbb{S}^1\cong\mathbb{T}$.} $f$ (the superscript denotes the derivative with respect to that):
\begin{multline}
\label{eqn:Lagr}
L(x,y,z,x^{\prime},y^{\prime},z^{\prime},f)=\frac12\left((x^{\prime})^2+(y^{\prime})^2+(z^{\prime})^2\right)+xy^{\prime}-x^{\prime}y\\
+\frac{1}{1+\varepsilon\cos f}\left(\frac{1-\mu}{d_1}+\frac{\mu}{d_2}+\frac{1}{2}(x^2+y^2-z^2\varepsilon\cos f)\right)\;,
\end{multline}
where, explicitly, $d_1=\sqrt{(x+\mu)^2+y^2+z^2}$ and $d_2=\sqrt{(x-1+\mu)^2+y^2+z^2}$.\\
\indent The origin of the coordinate axes is now moved to one of the two singular positions and thus, as mentioned above in Section \ref{sec:thm}, we choose $P_2(x_2,y_2,z_2)$:
\begin{equation}
\label{eqn:trans2}
(x-x_2,y,z)=q\;.
\end{equation}
Then \eqref{eqn:Lagr} becomes:
\begin{multline}
\label{eqn:transLag}
\widetilde{L}(q,q^{\prime},f)=\frac12\Vert q^{\prime}\Vert^2+q^{\prime}\times(0,0,1)\cdot q\\
+\frac{1}{1+\varepsilon\cos f}\bigg[(1-\mu)\bigg(\frac{1}{\Vert q+(1,0,0)\Vert}
+q_1\bigg)\\
+\frac{\mu}{\Vert q\Vert}+\frac 12(q_1^2+q_2^2-q_3^2\varepsilon\cos f)\bigg]\;,
\end{multline}
where the addenda $q^{\prime}\times (0,0,1)\cdot (x_2,0,0)$ and $(1-\mu)^2/(2(1+\varepsilon\cos f))$ have been dropped because they do not contribute to the Lagrange equations.

\subsection{The space of redundant variables}
\label{subsec:redund}
By following \cite{GuzzoCardin:intCR3BP} and applying their argument to
the elliptic problem, we apply  the projection map defined by \eqref{eqn:KSprojrel} to 
the previous Lagrangian $\widetilde{L}$,  and we compute the function $\mathscr{L}(u,u^{\prime},f)$ exploiting the relationship:
\begin{equation*}
(q_1^{\prime},q_2^{\prime},q_3^{\prime},0)=2A(u)u^{\prime}-2(0,0,0,l(u,u^{\prime}))\;,
\end{equation*}
in which
\begin{equation}
\label{eqn:bil}
l(u,u^{\prime})=u_4u_1^{\prime}-u_3u_2^{\prime}+u_2u_3^{\prime}-u_1u_4^{\prime}
\end{equation}
is the bilinear form appearing in the usual KS regularization. We obtain:
\begin{multline}
\label{eqn:lagredun}
\mathscr{L}(u,u^{\prime},f)=\widetilde{L}\left(\pi(u),\frac{\partial\pi}{\partial u}(u)u^{\prime},f \right)=2\Vert u\Vert^2\Vert u^{\prime}\Vert^2-2l^2(u,u^{\prime})+b(u)\cdot u^{\prime}\\
+\frac{1}{1+\varepsilon\cos f}\bigg[(1-\mu)\bigg(\frac{1}{\Vert\pi(u)+(1,0,0)\Vert}+\pi_1(u)\bigg)\\
+\frac{\mu}{\Vert u\Vert^2}+\frac 12(\pi_1^2(u)+\pi_2^2(u)-\pi_3^2(u)\varepsilon\cos f)\bigg]\;,
\end{multline}
for $b(u)$ expressed as in \eqref{eqn:b}.\\
\indent The first task consists in proving the specific invariance of Lagrange equations under the transformation at issue. In practice, the solutions of Lagrange equations for $\mathscr{L}(u,u^{\prime},f)$, which we write using the operator notation:
\begin{equation}
\label{eqn:KSLoperu}
[\mathscr{L}]_i(u,u^{\prime},u^{\prime\prime},f)=\frac{\text{d}}{\text{d}f}\frac{\partial\mathscr{L}}{\partial u_i^{\prime}}-\frac{\partial\mathscr{L}}{\partial u_i}=0\;,\quad\forall i=1,2,3,4\;,
\end{equation}
have to be compared with the solutions of Lagrange equations for $\widetilde{L}(q,q^{\prime},f)$, denoted by:
\begin{equation}
\label{eqn:KSLoperq}
\big<\widetilde{L}\big>_i(q,q^{\prime},q^{\prime\prime},f)=\frac{\text{d}}{\text{d}f}\frac{\partial\widetilde{L}}{\partial q_j^{\prime}}-\frac{\partial\widetilde{L}}{\partial q_j}=0\;,\quad\forall j=1,2,3\;.
\end{equation}
With the following statement it turns out that this requirement is fulfilled as soon as the solution $u(f)\neq0$ for all $f\in\mathbb{T}$.
\begin{proposition}
	\label{prop:lagruq}
	If $u(f)$ is a solution of Lagrange equations associated to $\mathscr{L}(u,u^{\prime},f)$ with initial condition $u(0)\neq0$, then $q(f)=\pi(u(f))$ is a solution of Lagrange equations associated to $\widetilde{L}(q,q^{\prime},f)$ as soon as $u(f)\neq0$.
\end{proposition}
\begin{proof}
	For any smooth curve $u(f)$, reminding that:
	\begin{equation*}
	\mathscr{L}(u,u^{\prime},f)=\widetilde{L}\left(\pi(u),\frac{\partial\pi}{\partial u}(u)u^{\prime},f\right)\;,
	\end{equation*}
	as well as
	\begin{equation*}
	\frac{\partial q^{\prime}}{\partial u^{\prime}}=\frac{\partial \pi}{\partial u}\;,
	\end{equation*}
	one gets from the chain rule:
	\begin{equation*}
	\frac{\partial\mathscr{L}}{\partial u_i^{\prime}}=
	\sum_{j=1}^3\frac{\partial\widetilde{L}}{\partial q_j^{\prime}}\frac{\partial q'_j}{\partial u'_i}=\sum_{j=1}^3\frac{\partial\widetilde{L}}{\partial q_j^{\prime}}\frac{\partial\pi_j}{\partial u_i}\;
	\end{equation*}
	and
\begin{align*}
\frac{\text{d}}{\text{d}f}\frac{\partial\mathscr{L}}{\partial u_i^{\prime}}&=\sum_{j=1}^3\frac{\text{d}}{\text{d}f}\frac{\partial\widetilde{L}}{\partial q_j^{\prime}}\frac{\partial\pi_j}{\partial u_i}+\sum_{j=1}^3\frac{\partial\widetilde{L}}{\partial q_j^{\prime}}\frac{\text{d}}{\text{d}f}\frac{\partial\pi_j}{\partial u_i}\\
&=\sum_{j=1}^3\frac{\text{d}}{\text{d}f}\frac{\partial\widetilde{L}}{\partial q_j^{\prime}}\frac{\partial\pi_j}{\partial u_i}+\sum_{j=1}^3\frac{\partial\widetilde{L}}{\partial q_j^{\prime}}\sum_{k=1}^{4}\frac{\partial^2\pi_j}{\partial u_i\partial u_k}u_k^{\prime}\;,
\end{align*}
\begin{equation*}
\frac{\partial\mathscr{L}}{\partial u_i}=\sum_{j=1}^3\frac{\partial\widetilde{L}}{\partial q_j}\frac{\partial\pi_j}{\partial u_i}+\sum_{j=1}^3\frac{\partial\widetilde{L}}{\partial q_j^{\prime}}\sum_{k=1}^{4}\frac{\partial^2\pi_j}{\partial u_i\partial u_k}u_k^{\prime}\;,
\end{equation*}
for $i=1,2,3,4$. As a consequence we have:
\begin{equation*}
[\mathscr{L}](u(f),u^{\prime}(f),u^{\prime\prime}(f),f)=\left ({\partial \pi\over\partial u}(u(f))\right )^T\big<\widetilde{L}\big>\left(\pi(u(f)),\frac{\text{d}}{\text{d}f}\pi(u(f)),\frac{\text{d}^2}{\text{d}f^2}\pi(u(f)),f\right)\;
\end{equation*}
where $[\mathscr{L}]\in\mathbb{R}^4$, $\big<\widetilde{L}\big>\in\mathbb{R}^3$ are the vectors of components respectively $[\mathscr{L}]_i$, $\big<\widetilde{L}\big>_j$.\\
Since  by assumption $[\mathscr{L}]_i(u(f),u^{\prime}(f),u^{\prime\prime}(f),f)=0$, the vector
\begin{equation*}
\big<\widetilde{L}\big>\left(\pi(u(f)),\frac{\text{d}}{\text{d}f}\pi(u(f)),\frac{\text{d}^2}{\text{d}f^2}\pi(u(f)),f\right)
\end{equation*}
is for any $f$ in the kernel of the matrix $\left ({\partial \pi\over\partial u}(u(f))\right )^T$. We claim that the kernel of $\left ({\partial \pi\over \partial u}(u(f))\right )^T$ contains only $(0,0,0)$ if $u\ne 0$. In fact, an element $(\alpha,\beta,\eta)$ is in the kernel of $\left ({\partial \pi\over \partial u}(u(f))\right )^T$ if and only if its components satisfy the system:
\begin{equation*}
\begin{cases}
u_1\alpha+u_2\beta+u_3\eta=0\\
-u_2\alpha+u_1\beta+u_4\eta=0\\
-u_3\alpha-u_4\beta+u_1\eta=0\\
u_4\alpha-u_3\beta+u_2\eta=0
\end{cases},
\end{equation*}
which admits the unique solution $\alpha=\beta=\eta=0$ as long as at least one of the components of $u$ is different from zero.\\
This implies $\big<\widetilde{L}\big>\left(\pi(u(f)),\frac{\text{d}}{\text{d}f}\pi(u(f)),\frac{\text{d}^2}{\text{d}f^2}\pi(u(f)),f\right)=0$ and $q(f)=\pi(u(f))$ is a solution of the Lagrange equations of $\widetilde{L}$.
\end{proof}

\subsubsection*{The modified Lagrangian}
The second matter to tackle regards the Legendre transform (necessary to deduce in Subsection \ref{subsec:regul} the corresponding transformed Hamiltonian and then proceed with the development):
\begin{equation}
\label{eqn:KSlegendre}
\frac{\partial\mathscr{L}}{\partial u^{\prime}}=\left(\frac{\partial\mathscr{L}}{\partial u_1^{\prime}},\frac{\partial\mathscr{L}}{\partial u_2^{\prime}},\frac{\partial\mathscr{L}}{\partial u_3^{\prime}},\frac{\partial\mathscr{L}}{\partial u_4^{\prime}}\right)=4\Vert u\Vert^2u^{\prime}-4(\Omega u\cdot u^{\prime})\Omega u+b(u)\;,
\end{equation}
where 
\begin{equation}
\label{eqn:KSpermut}
\Omega=
\begin{pmatrix}
0 & 0 & 0 & 1\\
0 & 0 & -1 & 0\\
0 & 1 & 0 & 0\\
-1 & 0 & 0 & 0
\end{pmatrix}
\end{equation}
is an ad hoc permutation matrix coming from the bilinear form term ($l(u,u^{\prime})=\Omega u\cdot u^{\prime}$), which is not invertible with respect to the generalized velocities, because the Hessian matrix
\begin{equation}
\label{eqn:KShess}
\mathscr{H}_{u^{\prime}}=\left(\frac{\partial^2\mathscr{L}}{\partial u_i^{\prime}\partial u_j^{\prime}}\right)\;,\quad i,j\in\{1,2,3,4\}\;,
\end{equation}
is identically singular, indeed $\det\mathscr{H}_{u^{\prime}}=0$. To overcome the degeneracy we proceed as in \cite{GuzzoCardin:intCR3BP}:
it is profitable to change the Lagrangian just by adding two times the square of the bilinear form (so that $-2l^2$ vanishes): such artifice precisely allows to restore the invertibility, thereby:
\begin{multline}
\label{eqn:KSmodifL}
\mathcal{L}(u,u^{\prime},f)=\mathscr{L}(u,u^{\prime},f)+2l^2(u,u^{\prime})=2\Vert u\Vert^2\Vert u^{\prime}\Vert^2+b(u)\cdot u^{\prime}\\
+\frac{1}{1+\varepsilon\cos f}\bigg[(1-\mu)\bigg(\frac{1}{\Vert\pi(u)+(1,0,0)\Vert}+\pi_1(u)\bigg)\\
+\frac{\mu}{\Vert u\Vert^2}+\frac 12(\pi_1^2(u)+\pi_2^2(u)-\pi_3^2(u)\varepsilon\cos f)\bigg]
\end{multline}
is the modified Lagrangian and in fact, introducing the KS momenta $U=(U_1,U_2,U_3,U_4)$ conjugate to $u=(u_1,u_2,u_3,u_4)$, the relationship
\begin{equation}
\label{eqn:KSLegendre}
U=\frac{\partial\mathcal{L}}{\partial u^{\prime}}(u,u^{\prime})=4\Vert u\Vert^2u^{\prime}+b(u)
\end{equation}
is non-degenerate (thus invertible) in $u^{\prime}$ for $u\neq0$.

\subsubsection*{Rotational invariance of the modified Lagrangian}
The sum of the quadratic expression $2l^2(u,u^{\prime})$ of course alters $\mathscr{L}(u,u^{\prime},f)$ and again one has to make sure that such action is legitimized under appropriate conditions (until now $u(f)\neq 0$ always). Let then the investigation begin by realizing the well known
remarkable symmetry property of the KS transformation. 
\begin{proposition}
	\label{prop:lagrrotinv}
	The modified Lagrangian $\mathcal{L}(u,u^{\prime},f)$ is invariant under the one-parameter family of transformations involving the redundant coordinates:
	\begin{align}
	\label{eqn:KSinvrot}
	\begin{split}
	S_{\theta}\colon\mathbb{R}^4&\longrightarrow\mathbb{R}^4\\
	u&\longmapsto S_{\theta}u
	\end{split}\;,
	\end{align}
	where $S_{\theta}\in SO(4)$ is the four-dimensional rotation matrix
	\begin{equation}
	\label{eqn:KSrotmat}
	S_{\theta}=
	\begin{pmatrix}
	\cos\theta & 0 & 0& -\sin\theta\\
	0 & \cos\theta & \sin\theta & 0\\
	0 & -\sin\theta & \cos\theta & 0\\
	\sin\theta & 0 & 0 & \cos\theta
	\end{pmatrix}\;,
	\end{equation}
	whose orbits define the fibers of the projection $\pi$, i.e., $\pi(S_{\theta}u)=\pi(u)$ for all $\theta\in\mathbb{T}$. More precisely:
	\begin{equation}
	\label{eqn:KSinvlag}
	\mathcal{L}(S_{\theta}u,S_{\theta}u^{\prime},f)=\mathcal{L}(u,u^{\prime},f)\;.
	\end{equation}
\end{proposition}
\begin{proof}
	This is a property of the KS transformation, which is the same for the CR3BP and for the ER3BP. Thus, for the proof we refer to \cite{GuzzoCardin:intCR3BP}, Section \ref{sec:thm}.
\end{proof}
\noindent This fact implies that there exists, by Noether's theorem, a conserved quantity:
\begin{equation*}
J(u,u^{\prime})=\frac{\partial\mathscr{L}}{\partial u^{\prime}}\cdot\frac{\text{d}}{\text{d}\theta}S_{\theta}u\bigg|_{\theta=0}=-4\Vert u\Vert^2l(u,u^{\prime})-b(u)\cdot\Omega u=-4\Vert u\Vert^2l(u,u^{\prime})
\end{equation*}
which is an autonomous first integral for the Lagrangian $\mathcal{L}$. For convenience the final constant of motion is given by:
\begin{equation}
\label{KS:firstint}
\mathcal{J}(u,u^{\prime})=\Vert u\Vert^2l(u,u^{\prime}).
\end{equation}
If the bilinear form is cleverly zeroed out at $f=0$ (by proper initial conditions), it will keep taking zero value for further $f$ (since we only consider time intervals such that $\Vert u\Vert\neq0$), so the extra factor $2l^2$ would become a vanishing contribution to the Lagrange equations.\\
According to such idea, the bilinear form assumes the meaning of constraint to be respected along the motion and the final claim, whose proof in the proposition below resolves completely the issue, is that Lagrange equations associated to $\mathcal{L}$ have the same solutions of the Lagrange equations associated to $\mathscr{L}$.
\begin{proposition}
	\label{prop:modlagreqq}
	If $u(f)$ is a solution of the Lagrange equations of $\mathcal{L}(u,u^{\prime},f)$ with initial data $u(0)$, $u^{\prime}(0)$ satisfying $u(0)\neq0$ and $l(u(0),u^{\prime}(0))=0$, then it is also a solution of the Lagrange equations of $\mathscr{L}(u,u^{\prime},f)$ as long as $u(f)\neq0$.
\end{proposition}
\begin{proof}
	Consider a solution $u(f)$ of the $\mathcal{L}$-equations with $u(0)\neq0 $ and $l(u(0),\allowbreak u^{\prime}(0))=0$. As long as $u(f)\neq0$, by \eqref{KS:firstint}, $l(u(f),u^{\prime}(f))=0$.	Moreover:
	\begin{equation*}
	\frac{\text{d}}{\text{d}f}l(u(f),u^{\prime}(f))=l(u^{\prime}(f),u^{\prime}(f))+l(u(f),u^{\prime\prime}(f))=l(u(f),u^{\prime\prime}(f))\;.
	\end{equation*}
	Now $u(f)$ solves the Lagrange equations for $\mathscr{L}$ too, in fact, referring to the previous notation \eqref{eqn:KSLoperu}, for any $i=1,2,3,4$:
	\begin{align*}
	[\mathscr{L}]_i&=[\mathcal{L}-2l^2]_i
	=[\mathcal{L}]_i-2\left(\frac{\text{d}}{\text{d}f}\frac{\partial}{\partial u_i^{\prime}}l^2(u,u^{\prime})-\frac{\partial}{\partial u_i}l^2(u,u^{\prime})\right)\\
	&=[\mathcal{L}]_i-4\left[\frac{\text{d}}{\text{d}f}\left(l(u,u^{\prime})\frac{\partial}{\partial u_i^{\prime}}l(u,u^{\prime})\right)-l(u,u^{\prime})\frac{\partial}{\partial u_i}l(u,u^{\prime})\right]
	\end{align*}
	and when evaluated along the curve $u(f)$:
	\begin{multline*}
	[\mathscr{L}]_i(u(f),u^{\prime}(f),u^{\prime\prime}(f),f)=[\mathcal{L}]_i(u(f),u^{\prime}(f),u^{\prime\prime}(f),f)\\
	-4\bigg(l(u(f),u^{\prime\prime}(f))\frac{\partial}{\partial u^{\prime}_i}l(u(f),u^{\prime}(f))+l(u(f),u^{\prime}(f))\frac{\text{d}}{\text{d}f}\frac{\partial}{\partial u_i^{\prime}}l(u(f),u^{\prime}(f))\\
	-l(u(f),u^{\prime}(f))\frac{\partial}{\partial u_i}l(u(f),u^{\prime}(f))\bigg)=0\;,
	\end{multline*}
	owing to $l(u(f),u^{\prime\prime}(f))=l(u(f),u^{\prime}(f))=0$.
\end{proof}

\subsection{The regularized Hamiltonian}
\label{subsec:regul}
The corresponding singular Hamiltonian enters now by performing the Legendre transform:
\begin{equation}
\label{eqn:LegHam}
\mathscr{K}(u,U,f)=U\cdot g(u,U)-\mathcal{L}(u,g(u,U),f)\;,
\end{equation}
where
\begin{equation*}
u^{\prime}=g(u,U)=\frac{U-b(u)}{4\Vert u\Vert^2}
\end{equation*}
is the inverse of $U$ with respect to $u^{\prime}$; more explicitly
we have:
\begin{multline}
\label{eqn:modhamnaut}
\mathscr{K}(u,U,f)=\frac{1}{8\Vert u\Vert^2}\Vert U-b(u)\Vert^2\\
-\frac{1}{1+\varepsilon\cos f}\bigg[(1-\mu)\bigg(\frac{1}{\Vert\pi(u)+(1,0,0)\Vert}+\pi_1(u)\bigg)\\
+\frac{\mu}{\Vert u\Vert^2}+\frac 12(\pi_1(u)^2+\pi_2(u)^2-\pi_3(u)^2\varepsilon\cos f)\bigg]\;
\end{multline}
and the bilinear equality $l(u,u^{\prime})=0$ straightforwardly translates in the Hamiltonian formalism as $l(u,U)=0$ for $u\neq0$, because
\begin{equation*}
l(u,u^{\prime})=l(u,g(u,U))=\frac{1}{4\Vert u\Vert^2}l(u,U)-\frac{1}{4\Vert u\Vert^2}l(u,b(u))\;,
\end{equation*}
but $l(u,b(u))=\Omega u\cdot b(u)=0$ identically, hence
we have $l(u,u^{\prime})=0$ if and only if $l(u,U)=0$.\\
\indent With all this in hands it is useful to work with an autonomous extension of the transformed Hamiltonian $\mathscr{K}$. So we append one more degree of freedom to form the extended phase space $T^*((\mathbb{R}^4\setminus\mathscr{C})\times\mathbb{T})$, where
\begin{equation}
	\label{eqn:collset}
        \mathscr{C}=\{(0,0,0,0)\}\cup \{\left(0,u_2,u_3,0\right)\colon u_2^2+u_3^2=1\}
\end{equation}
is the collision set in KS coordinates, with the extra couple of variables $(\phi,\Phi)\in\mathbb{T}\times\mathbb{R}$ and standard symplectic form $\sum_{i=1}^{4}\text{d}u_i\wedge \text{d}U_i+\text{d}\phi\wedge \text{d}\Phi$, in such a way to build the autonomous transformed Hamiltonian:
\begin{equation}\label{eqn:autK}
\widehat{\mathscr{K}}\,(u,\phi,U,\Phi)=\mathscr{K}(u,U,\phi)+\Phi\;,
\end{equation}
and consider the solutions $u(f),\phi(f),U(f),\Phi(f)$ of the Hamilton equations of (\ref{eqn:autK}) such that, for given initial value $f_0$ of the true anomaly, satisfy:
\begin{equation*}
u(f_0)=u_0\;,\quad\phi(f_0)=f_0\;,\quad U(f_0)=U_0\;,\quad\Phi(f_0)=-\mathscr{K}(u_0,U_0,f_0)\;.
\end{equation*}
At this point we perform a rescaling similar to the one in the Levi-Civita
regularization, and define the regularized Hamiltonian: 
\begin{multline}
\label{eqn:hamreg}
\mathcal{K}(u,\phi,U,\Phi)
=\Vert u\Vert^2\widehat{\mathscr{K}}\,(u,\phi,U,\Phi)
=\frac{1}{8}\Vert U-b(u)\Vert^2\\
-\frac{1}{1+\varepsilon\cos \phi}\bigg[(1-\mu)\Vert u\Vert^2\bigg(\frac{1}{\Vert\pi(u)+(1,0,0)\Vert}+\pi_1(u)\bigg)+\mu\\
+\frac 12\Vert u\Vert^2(\pi_1^2(u)+\pi_2^2(u)-\pi_3^2(u)\varepsilon\cos \phi)+\frac{(1-\mu)^2}{2}\Vert u\Vert^2\bigg]+\Phi\Vert u\Vert^2\;.
\end{multline}
\begin{remark}
	\label{rem:Kreg}
	\phantom{text}
		\begin{itemize}
		\item [-] For $\varepsilon=0$, the action $\Phi$ is a constant of
		motion and the Hamiltonian \eqref{eqn:hamreg} is identical to the KS Hamiltonian of the CR3BP, as represented in \cite{GuzzoCardin:intCR3BP} with $\Phi=-E$.
		\item[-] $\mathcal{K}(u,\phi,U,\Phi)$ is invariant under the same one-parameter family of transformations defined by \eqref{eqn:KSinvrot} and \eqref{eqn:KSrotmat}, hence $\mathcal{J}(u,g(u,U))=l(u,U)$ is a first integral also for the Hamilton equations of $\mathcal{K}(u,\phi,U,\Phi)$. 
		\item[-] Hamiltonian \eqref{eqn:hamreg} is regular at $u=0$. 
	\end{itemize}
\end{remark}

\subsection{Projection of the solutions of the regularized Hamiltonian}
\label{subsec:conj}

Let us prove that the solutions of the Hamilton equations of
the regularized Hamiltonian (\ref{eqn:hamreg}) project
on the Hamilton solutions of the Hamiltonian (\ref{eqn:ER3BP3Dhamfintro})
of the ER3BP. Similarly to the classic LC and KS techniques we need an independent variable redefinition, which for the ER3BP is obtained by introducing the fictitious true anomaly $s$ such that:
\begin{equation}
\label{eqn:fictruean}
s^{\prime}(f)=\frac{1}{\Vert u(f)\Vert^2}\;,\quad\quad s(f_0)=0\;,
\end{equation}
whose inverse is precisely $\partial\mathcal{K}/\partial\Phi$. Thereby we state our result.

\begin{theorem}
	The solutions $(u(s),\phi(s),U(s),\Phi(s))$ of Hamilton equations related to $\mathcal{K}(u,\phi,\allowbreak U,\Phi)$ with initial conditions satisfying:
	\begin{enumerate}[label=(\roman*)]
		\item $u(0)\neq0$,
		\item $l(u(0),U(0))=0$,
		\item $\mathcal{K}(u(0),\phi(0),U(0),\Phi(0))=0$, $\phi(0)=f_0$,
	\end{enumerate}
project, for $s$ in a neighborhood of $s=0$, via the true anomaly reparametrization:
	\begin{equation}
	\label{eqn:fictrueanfofs}
	f(s)=f_0+\int_{0}^s\Vert u(\sigma)\Vert^2\text{d}\sigma\;,
	\end{equation}
	the transformation \eqref{eqn:KSprojrel} and the translation \eqref{eqn:trans}, onto solutions $(x(f),y(f),z(f),p_1(f),\allowbreak p_2(f),p_3(f))$ of the Hamilton equations of Hamiltonian \eqref{eqn:ER3BP3Dhamfintro}. 
\end{theorem}

\begin{proof}
In light of what already derived in the previous subsections, we only need to prove the equivalence between the solutions associated to the transformed $\widehat{\mathscr{K}}$ and the regularized $\mathcal{K}$. Given the initial conditions $u_0,U_0,f_0$, let us consider the solution $(\widetilde{u}(s),\widetilde{\phi}(s),\widetilde{U}(s),\allowbreak\widetilde{\Phi}(s))$ of the Hamilton equations of $\mathcal{K}$ with:
  $$
  \widetilde{u}(0)=u_0\;,\ \ \widetilde{U}(0)=U_0\; ,\ \
  \widetilde{\phi}(0)=f_0\; ,\ \ \widetilde{\Phi}(0)=-\mathscr{K}(u_0,U_0,f_0)\;,
  $$
and $s$ in a neighborhood of $s=0$ such that $\Vert \widetilde{u}(s)\Vert >0$; in particular we have:
  $$
  \mathcal{K}(\widetilde{u}(s),\widetilde{\phi}(s),\widetilde{U}(s),\widetilde{\Phi}(s))=0
  $$
for all $s$. Next, consider:
\begin{equation*}
f(s)=f_0+\int_{0}^s\Vert \widetilde{u}(\sigma)\Vert^2\text{d}\sigma\;,
\end{equation*}  
which is invertible (since in the neighborhood of $s=0$ we have
$\Vert \widetilde{u}(s)\Vert >0$), and $(u(f),\phi(f),U(f),\Phi(f))$ defined by:
$$
u(f)= \widetilde u(s(f))\; ,\ \ U(f)= \widetilde U(s(f))\; ,\ \ \phi(f)=\widetilde \phi(s(f))\; ,\ \ \Phi(f)=
\widetilde \Phi(s(f))\;.
$$
We claim that $(u(f),\phi(f),U(f),\Phi(f))$ are the solutions
of the Hamilton equations of $\widehat{\mathscr{K}}$ with initial
conditions $(u(f_0),\phi(f_0),U(f_0),\Phi(f_0))= (u_0,f_0,U_0,
-\mathscr{K}(u_0,U_0,f_0))$. In fact, we have:
\begin{align*}
{\text{d}u_i\over \text{d}f}= {\text{d}s\over \text{d}f}\ { \text{d}\widetilde u_i\over \text{d} s}_{\Big\vert s= s(f)}&={1\over \Vert \widetilde{u}(s(f))\Vert^2}
\left [  \frac{\partial}{\partial U_i}\left(\Vert u\Vert^2\widehat{\mathscr{K}}\;\right)\right ]_{\Big\vert  (u,U,\phi,\Phi)= (\widetilde{u}(s),\widetilde{\phi}(s),\widetilde{U}(s),\widetilde{\Phi}(s)), s= s(f)}\\
&=\left [  \frac{\partial}{\partial U_i}\widehat{\mathscr{K}}\;\right ]
(\widetilde{u}(s(f)),\widetilde{\phi}(s(f)),\widetilde{U}(s(f)),\widetilde{\Phi}(s(f)))\\
&=\left [  \frac{\partial}{\partial U_i}\widehat{\mathscr{K}}\;\right ]
(u(f),\phi(f),U(f),\Phi(f))\;,
\end{align*}
as well as:
\begin{align*}
{\text{d} \phi \over \text{d}f}= {\text{d}s\over \text{d}f}\ { \text{d}\widetilde \phi\over \text{d} s}_{\Big\vert s= s(f)}&={1\over \Vert \widetilde{u}(s(f))\Vert^2}  
\left [  \frac{\partial}{\partial \Phi}\left(\Vert u\Vert^2\widehat{\mathscr{K}}\;\right)\right ]_{\Big\vert  (u,U,\phi,\Phi)= (\tilde{u}(s),\tilde{\phi}(s),\widetilde{U}(s),\widetilde{\Phi}(s)), s= s(f)}\\
&=\frac{\partial \widehat{\mathscr{K}}}{\partial \Phi}=1
\end{align*}
and:
\begin{align*}
{\text{d} U_i\over \text{d}f}= {\text{d}s\over \text{d}f}\ { \text{d}\widetilde U_i\over \text{d} s}_{\Big\vert s= s(f)}&=-{1\over \Vert \widetilde{u}(s(f))\Vert^2}  
\left [  \frac{\partial}{\partial u_i}\left(\Vert u\Vert^2\widehat{\mathscr{K}}\;\right)\right ]_{\Big\vert  (u,U,\phi,\Phi)= (\tilde{u}(s),\tilde{\phi}(s),\widetilde{U}(s),\widetilde{\Phi}(s)), s= s(f)}\\
&=-\left [  \frac{\partial}{\partial u_i}\widehat{\mathscr{K}}\;\right ]
(\tilde{u}(s(f)),\tilde{\phi}(s(f)),\widetilde{U}(s(f)),\widetilde{\Phi}(s(f)))\\
&=-\left [  \frac{\partial}{\partial u_i}\widehat{\mathscr{K}}\;\right ]
(u(f),\phi(f),U(f),\Phi(f))\;,
\end{align*}
where to obtain the second equality we used $\widehat{\mathscr{K}}(\widetilde{u}(s),\widetilde{\phi}(s),\widetilde{U}(s),\widetilde{\Phi}(s))=0$. Finally, we also have:
\begin{align*}
{\text{d} \Phi\over \text{d}f}= {\text{d}s\over \text{d}f}\ { \text{d}\widetilde \Phi\over \text{d} s}_{\Big\vert s= s(f)}&=-{1\over \Vert \widetilde{u}(s(f))\Vert^2}  
\left [  \frac{\partial}{\partial \phi}\left(\Vert u\Vert^2\widehat{\mathscr{K}}\;\right)\right ]_{\Big\vert  (u,U,\phi,\Phi)= (\tilde{u}(s),\tilde{\phi}(s),\widetilde{U}(s),\widetilde{\Phi}(s)), s= s(f)}\\
&=-\left [  \frac{\partial}{\partial \phi}\widehat{\mathscr{K}}\;\right ]
(\tilde{u}(s(f)),\tilde{\phi}(s(f)),\widetilde{U}(s(f)),\widetilde{\Phi}(s(f)))\\
&=-\left [  \frac{\partial}{\partial \phi}\widehat{\mathscr{K}}\;\right ]
(u(f),\phi(f),U(f),\Phi(f))\;.
\end{align*}
\end{proof}

\begin{remark}
    \label{rem:Ksympred}
The reason for the success of this Hamiltonian regularization
 is, as for the spatial CR3BP, the possibility to exploit the symmetry
  presented in Proposition \ref{prop:lagrrotinv} in the framework of symplectic reductions
 (see \cite{marsden1974reduction,meyer1973symmetries,saha,ZhaoRCD}) in the
 10-dimensional phase-space $(u,\phi,U,\Phi)$.
\end{remark}

\section{Fast close encounters in the ER3BP}
\label{sec:fastcloseencounters}

The regularized Hamiltonian (\ref{eqn:hamregintro}) has the representation:
\begin{equation}
	\label{eqn:KfunR}
\mathcal{K}=\frac{1}{8}\Vert U-b(u)\Vert^2
-\Vert {u}\Vert^2\left [-\Phi + {3-4 \mu +\mu^2\over 2 (1+\varepsilon \cos\phi)}
\right ]-{\mu \over 1+\varepsilon \cos\phi} +\mathcal{R}_6(u,\phi)\;, 
\end{equation}
where
\begin{eqnarray}
\mathcal{R}_6 &=& -\frac{1}{1+\varepsilon\cos\phi}\bigg[(1-\mu)\Vert u\Vert^2\bigg(\frac{1}{\Vert\pi(u)+(1,0,0)\Vert}+\pi_1(u)-1\bigg)\cr
&+&\frac 12\Vert u\Vert^2\left(\pi_1^2(u)+\pi_2^2(u)-\pi_3^2(u)\varepsilon\cos \phi\right)\bigg ]
\end{eqnarray}
has Taylor expansion with respect to the variables $u$ starting
at order 6. In fact, the Taylor expansion of $1/\Norm{\pi +(1,0,0)}+\pi_1 -1$
with respect to $\pi_1,\pi_2,\pi_3$ starts at order $2$, and the $\pi_j$ are
quadratic functions of the $u_i$. We provide 
upper bounds to the remainder
$\mathcal{R}_6$ and its derivatives in the ball:
\begin{equation}
B(\mu^{1\over 6})= \{u\in {\Bbb R}^4\colon \Norm{u}< \mu^{1\over 6}\}\;,
\label{usphere}
\end{equation}
projecting to the ball in the Cartesian
space:
\begin{equation}
B(\mu^{1\over 3})= \{q\in {\Bbb R}^3\colon \Norm{q}< \mu^{1\over 3}\}\;.
\label{xyzsphere}
\end{equation}
 We will refer to both
 spheres (\ref{usphere}) and (\ref{xyzsphere}) as the Hill's sphere (notice that
 the terminology is different from the conventional one by a numerical 
 factor in the radius of the sphere).

 Below, we denote by $\mu_0$, by $c_1,c_2,\ldots >0$ and by $\gamma_1,\gamma_2,\ldots >0$ constants which are independent of $\mu$ and $\varepsilon$, as well as on the parameter $\Gamma_0$ which will be later introduced to characterize each close encounter. 
\vskip 0.4 cm
\noindent
    {\bf (i)} There exist constants
$c_1,c_2,c_3,c_4>0$ and $\mu_0 \in (0,1/10)$ such that, 
for any $\mu\leq \mu_0$, $\varepsilon \in (0,1)$, $u\in B(\mu_0^{1\over 6})$ and $\phi \in [0,2\pi]$, we have:
\begin{eqnarray}
  \norm{\mathcal{R}_6(u,\phi)} &\leq& {c_1\over 1-\varepsilon} \Norm{u}^6\label{c1}\;,\\
  \Norm{b(u)} &\leq & c_2 \Norm{u}^{3} \label{c2}\;,\\
 \sum_{i} \norm{\sum_j u_j \left ({\partial b_i \over \partial u_j}- {\partial b_j \over \partial u_i}\right )}
    &\leq& c_3 \Norm{u}^{3}\; \label{c3}\;,\\
\norm{\sum_{i} u_{i} {\partial \mathcal{R}_6\over \partial u_{i}}} &\leq& {c_4\over 1-\varepsilon} \Norm{u}^6\;.\label{c4}
  \end{eqnarray}
The proof is reported in Subsection \ref{c1c2c3c4}.
\vskip 0.4 cm
\noindent
{\bf (ii)} We consider solutions $q(f)$ of the ER3BP transiting through the
Hill's sphere of the secondary body $P_2$ in the $f$-interval $[f_1,f_2]$, i.e.,
$\Norm{q(f_1)}=\Norm{q(f_2)}=\mu^{1/3}$ and $0<\Norm{q(f)}<\mu^{1\over 3}$ for all $f\in (f_1,f_2)$, and denote by $(u(s),\phi(s),U(s),\Phi(s))$ a solution of the Hamilton equations of the regularized Hamiltonian ${\cal K}$ 
satisfying the hypotheses of Subsection \ref{subsec:conj} such that
$\pi(u(s))=q(f(s))$ for all $s\in[s_1,s_2]$ ($s_1,s_2$ corresponding to $f_1,f_2$). If $\mu \leq \mu_0$, there exists a constant $\cvarGamma>0$ such that, during the transit in the Hill's sphere,  the variation of the parameter $\Gamma_s$ introduced in Section \ref{sec:intro} satisfies:
\begin{equation}
\norm{\frac{\text{d}\Gamma_s}{\text{d}s}}  \le \cvarGamma \frac{\varepsilon\mu  }{(1-\varepsilon)^2}\;.
\label{varGamma}
\end{equation}
The proof is reported in  Subsection \ref{Gamma'section}. \\
A direct consequence of (\ref{varGamma}) is: by considering transits with $\Gamma_0\coloneqq\Gamma_{s_1} >0$ and proper time intervals $\Delta s>0$ satisfying:
\begin{equation}
\Delta s   \frac{\varepsilon\mu  }{(1-\varepsilon)^2} \le {\Gamma_0\over 2 \cvarGamma}\;,
\label{estimate-transit}
\end{equation}
for all $s \in [s_1,s_1+\Delta s]\cap [s_1,s_2]$ we have $\Gamma_s \in  \left [\Gamma_0/2,3\Gamma_0/2\right ]$; in particular since  
$\Gamma_s\geq \Gamma_0/2>0$, the matrix ${\cal X}_0$ is hyperbolic.  

\vskip 0.2 cm
\noindent
    {\bf (iii)} For all $s \in [s_1,s_1+\Delta s]\cap [s_1,s_2]$,
    using ${\cal K}(u(s),\phi(s),U(s),\Phi(s))=0$ and \eqref{c1}, we obtain from \eqref{eqn:KfunR}:
$$
\frac{1}{8}\Vert U-b(u)\Vert^2
= \Gamma_s \Vert {u}\Vert^2+ {\mu \over 1+\varepsilon \cos\phi} -\mathcal{R}_6(u,\phi)
  \leq  {3\over 2}\Gamma_0 \Norm{u}^2 +
        {\mu \over 1-\varepsilon}+ {c_1\over 1-\varepsilon}\Norm{u}^6
        $$
and therefore there exists $\cU >0$ such that:
\begin{equation}
	\label{eqn:Uminb}
\Vert U-b(u)\Vert \leq \cU  \sqrt{\Gamma_0\Norm{u}^2 + {\mu \over 1-\varepsilon}}\;;
\end{equation}
by Cauchy-Schwarz inequality $\Norm{U-b}\ge\norm{\Norm{U}-\Norm{b}}$, thus, using \eqref{c2}:
$$
\Vert U\Vert \leq  \cU \sqrt{\Gamma_0\Norm{u}^2 + {\mu \over 1-\varepsilon}}
+c_2 \Norm{u}^3\;.
$$
Consequently,  $\Norm{ U(s)},\Norm{u(s)}$ are bounded by quantities which,
for $s \in [s_1,s_1+\Delta s]\cap [s_1,s_2]$, are small with $\mu^{1\over 6}$.
\vskip 0.4 cm 

\vskip 0.2 cm
\noindent
    {\bf (iv)} It remains to show that, with a suitable smallness condition on $\mu$, each transit in the Hill's sphere with $\Gamma_0>0$ occurs  in proper time intervals $\Delta s$ satisfying (\ref{estimate-transit}). First, we prove that there exists $c_7>0$ (indepedent  on $\mu,\varepsilon,\Gamma_0$) such that if $\mu$ satisfies:
\begin{equation}
\mu  \leq c_7 (1-\varepsilon)^2\Gamma_0^{3\over 2}  
\label{thrmu}
\end{equation}
and $\Delta s$ satisfies  (\ref{estimate-transit}), 
for all $s \in [s_1,s_1+\Delta s]\cap [s_1,s_2]$ we have:
\begin{equation}
\frac{\text{d}^2\rho}{\text{d}s^2}\geq  {\Gamma_0\over 2} \rho + {\mu \over 1+\varepsilon}\;, 
\label{rho''}
\end{equation}
where $\rho(s)\coloneqq \Norm{u(s)}^2$ (the proof is reported in Subsection \ref{rho''section}). 
\vskip 0.4 cm
\noindent
       {\bf (v)} Finally, we prove that if $\mu$ satisfies the
       inequality:
      \begin{equation}
\mu  < c (1-\varepsilon)^6\Gamma_0^{3\over 2}  
\label{thrmu2}
\end{equation} 
      with suitable $c>0$, the motion exits from the Hill's sphere
      within a proper time interval satisfying  (\ref{estimate-transit}).
      \vskip 0.4 cm
We first restrict the choice of $c$ in the interval $0<c\leq c_7$, so that       (\ref{thrmu2})  is stronger than  (\ref{thrmu}), i.e., the right-hand side of the former is smaller or equal than the one of the latter. As long as $\Delta s$ satisfies (\ref{estimate-transit})
and $\mu$ satisfies (\ref{thrmu2}), 
$\rho(s)$ satisfies (\ref{rho''}) and consequently we have:
\begin{eqnarray}
\rho(s)-\rho(s_1) &=& \frac{\text{d}\rho}{\text{d}s}(s_1)(s-s_1) + \int_{s_1}^s \left(\int_{s_1}^\tau \frac{\text{d}^2\rho}{\text{d}\sigma^2}(\sigma)\text{d}\sigma \right)\text{d}\tau\cr
& \geq&  \frac{\text{d}\rho}{\text{d}s}(s_1)(s-s_1) + \int_{s_1}^s\left(\int_{s_1}^\tau
            \left ({\Gamma_0\over 2}
            \rho(\sigma)+{\mu \over 1+\varepsilon}\right ) 
            \text{d}\sigma\right) \text{d}\tau\;.
            \label{rhoineq}
            \end{eqnarray}
            Since $\mu$ satisfies (\ref{thrmu}) and $\Norm{u(s_1)}=\mu^{\frac16}$, combining Cauchy-Schwarz inequality, Hamilton equations and \eqref{eqn:Uminb} we have:
            \begin{align*}
            \norm{\frac{\text{d}\rho}{\text{d}s}(s_1)}&\leq 2 \Norm{u(s_1)}\Norm{\frac{\text{d}u}{\text{d}s}(s_1)} =
            {1\over 2}\mu^{1\over 6}\Norm{U(s_1)-b(u(s_1))}\\
            &\leq {\cU\over 2} \mu^{1\over 6}\sqrt{\Gamma_0 \mu^{1\over 3}+{\mu
           \over 1-\varepsilon}} \leq c_8
       \mu^{1\over 3}\sqrt{\Gamma_0}\;,
       \end{align*}
       with $c_8$ depending only on $\cU,c_7$. 
       Therefore, from (\ref{rhoineq}) we obtain:
       \begin{equation}
       \rho(s)-\rho(s_1) > - c_8 \mu^{1\over 3}\sqrt{\Gamma_0}(s-s_1)
       +{\mu \over 2(1+\varepsilon)}(s-s_1)^2\;.
       \label{exit}
       \end{equation}
       Let us now consider from \eqref{exit} the inequality:
       \begin{equation}
       - c_8 \mu^{1\over 3}\sqrt{\Gamma_0}\omega
       + {\mu \over 2(1+\varepsilon)}\omega^2 > 0\;,
       \label{wineq}
       \end{equation}
       which is satisfied in particular by all $\omega>\omega_1>0$, with:
       $$
       \omega_1\coloneqq  2(1+\varepsilon)c_8 {\sqrt{\Gamma_0} \over  \mu^{2\over 3}}\;.
    $$
      Since:
       $$
       \omega_1 <  4c_8 {\sqrt{\Gamma_0}\over \mu^{2\over 3}}\eqqcolon\omega_2\;,
       $$
       inequality (\ref{wineq}) is satisfied
       in particular by all $\omega \geq \omega_2$. If $\mu$ satisfies the further inequality:
       \begin{equation}
         {\varepsilon \mu^{1\over 3}\over (1-\varepsilon)^2} \le {\sqrt{\Gamma_0}
           \over 8 c_5 c_8}\;,
         \label{ineqq}
       \end{equation}
    then the choice for $\Delta s=\omega_2$ satisfies inequality (\ref{estimate-transit}).
    Now we can state that $s_2 \in [s_1,s_1+\omega_2)$. In fact,
      if $s_1+\omega_2 \leq s_2$, then 
      inequality (\ref{exit}) is valid also for $s=s_1+\omega_2$:
       $$
       \rho(s_1+\omega_2)> \rho(s_1)=\mu^{1\over 6}\;,
       $$
       which contradicts the hypothesis $u(s_1+\omega_2)\in B(\mu^{1\over 6})$.
       Therefore, the motion exits from the Hill's sphere at a proper
    time $s \in (s_1,s_1+\omega_2)$.\\
Finally we notice that (\ref{ineqq}) is satisfied if:
    \begin{equation}
    \mu \leq c_9(1    -\varepsilon)^6 \Gamma_0^{3\over 2}\;.
    \label{thr3}
    \end{equation}
    with suitable constant $c_9$, 
    and therefore both inequalities (\ref{thr3}) and (\ref{thrmu})
    are satisfied by the $\mu$ satisfying (\ref{thrmu2}),
    provided that $c<\min\{c_7,c_9\}$.
    \vskip 0.4 cm
    \noindent
        {\bf (vi)} Therefore, we have proved that if $\mu$ satisfies
        a smallness condition, the matrix ${\cal X}_0$ is hyperbolic
        during the transit in the Hill's sphere. Moreover, if the Hill's
        radius is suitably small (the smallness condition depending possibly
        on $\Gamma_0$), which is obtained by introducing another smallness condition on $\mu$, then also the matrix ${\cal X}$ is hyperbolic
during the transit in the Hill's sphere. Indeed, in virtue of \eqref{eqn:matX}, \eqref{eqn:matX0} and the above bounds, the eigenvalues of $\mathcal{X}$ differ from those of $\mathcal{X}_0$ by terms of order $\mu^{1/3}$. 

\subsection{Proof of inequalities (\ref{c1}),(\ref{c2}),(\ref{c3}),(\ref{c4})}\label{c1c2c3c4}

Inequalities  (\ref{c2}) and (\ref{c3}) follow immediately from the
fact that the vector $b(u)$ is cubic in the $u_i$. Indeed:
\[
\Norm{b}^2=\norm{b\cdot b}=4\Norm{u}^2\norm{Au\cdot \Lambda^2 Au}\le4\Norm{u}^6
\]
thanks to $\Norm{\Lambda^2}=1$ (operator norm induced by the Euclidean vector norm), \eqref{eqn:Aprop} and Cauchy-Schwarz inequality, so it suffices that $c_2\ge2$. Consequently, $|\partial b_i/\partial u_j|\le c_*\Norm{u}^2$, $c_*>0$, thus (\ref{c3}) holds.\\ 
\indent To prove inequalities (\ref{c1}) and (\ref{c4}) we first replace in $\mathcal{R}_6$ the only non-polynomial term  with its Taylor expansion of order $1$ in the $\pi_1,\pi_2,\pi_3$ with remainder $\Delta_4$ (of order $2$ in the $\pi_j$, of order $4$ in the $u_i$):
\begin{equation}
{1\over {\Vert\pi+(1,0,0)\Vert}} = {1\over \sqrt{(\pi_1+1)^2+\pi_2^2+\pi_3^2}}=  1-\pi_1+\Delta_4\;,
\label{Delta4}
\end{equation}
and we obtain:
\begin{equation}
\mathcal{R}_6 = -\frac{\Vert u\Vert^2}{1+\varepsilon\cos\phi}\bigg[(1-\mu) \Delta_4
+\frac 12\left(\pi_1^2(u)+\pi_2^2(u)-\pi_3^2(u)\varepsilon\cos \phi\right)\bigg ]\;.
\end{equation}
Since for $\mu \leq \mu_0 < 1/10$  we have $d_2<(1/10)^{1/3}<1/2$, $d_1> 1/2$ in $B(\mu^{\frac13})$
as well as:
$$
1-d_1(\pi_1-1)= 1+d_1-d_1 \pi_1 \geq 1 +d_1(1-d_2)>1\;,
$$
we obtain:
$$
\norm{\Delta_4 } = \norm{{1\over d_1}+\pi_1-1} =
\norm{{1 +d_1(\pi_1-1)\over d_1}} 
=\norm{{1-d_1^2(\pi_1-1)^2 \over d_1 (1+d_1-d_1\pi_1)}}<
2\norm{1-d_1^2(\pi_1-1)^2}\;.
$$
Since $1-d_1^2(\pi_1-1)^2$ is a polynomial in $\pi_1,\pi_2,\pi_3$
with terms of  order $2,3,4$, there exists a constant $\cdelta$ such that:
\begin{equation}
\norm{\Delta_4}\leq \cdelta \Norm{u}^4\;.
\label{upperbounddelta4}
\end{equation}
Inequality (\ref{c1}) follows easily.\\
We proceed by estimating the derivatives of $\Delta_4$:
$$
{\partial \Delta_4 \over\partial u_j}=
\sum_{h=1}^3{\partial \Delta_4 \over\partial \pi_h}{\partial \pi_h\over
    \partial u_j}\;.
$$
Since the $\pi_h$ are quadratic function of the $u_j$, with $|\pi_h|\le\Norm{u}^2$, and since $1-d_1^6$ is a polynomial containing terms of least order in the $\pi_i$ equal to 1, 
there exists a positive constant $\gamma_*$ such that:
\begin{equation}
	\label{eqn:1minusd16}
	\norm{1-d_1^6}\le\gamma_*\Norm{u}^2
\end{equation}
and a positive constant $\cderdelta$ such that:
  \begin{equation*}
  \norm{{\partial \Delta_4 \over\partial \pi_2}{\partial \pi_2\over
    \partial u_j}} = \frac{\norm{ \pi_2}}{d_1^3} \norm{ {\partial \pi_2\over
    \partial u_j}} < 8  \norm{ \pi_2} \norm{ {\partial \pi_2\over
      \partial u_j}} \leq \cderdelta\Norm{u}^3\;,
  \end{equation*}
  \begin{equation*} 
  \norm{{\partial \Delta_4 \over\partial \pi_3}{\partial \pi_3\over
    \partial u_j}} = \frac{\norm{ \pi_3}}{d_1^3} \norm{ {\partial \pi_3\over
    \partial u_j}} < 8  \norm{ \pi_3} \norm{ {\partial \pi_3\over
      \partial u_j}} \leq \cderdelta\Norm{u}^3\;,
  \end{equation*}
  \begin{align*}
  	\norm{{\partial \Delta_4 \over\partial \pi_1}{\partial \pi_1\over
    \partial u_j}} =\frac{\norm{ 1+\pi_1-d_1^3}}{d_1^3}\norm{\partial \pi_1\over
    \partial u_j}&< 8 (\norm{\pi_1}+\norm{1-d_1^3})\norm{\partial \pi_1\over
    \partial u_j}= 8\left (\norm{\pi_1}+\frac{\norm{1-d_1^6}}{1+d_1^3}\right )\norm{\partial \pi_1\over
    \partial u_j}\\
  &< 8\left (\norm{\pi_1}+\norm{1-d_1^6}\right )\norm{\partial \pi_1\over
    \partial u_j}\leq \cderdelta\Norm{u}^3\;.
    \end{align*}
Inequality (\ref{c4}) follows straightforwardly.

\subsection{Proof of inequality (\ref{varGamma})}\label{Gamma'section}

In:
\begin{eqnarray}
  \frac{\text{d}\Gamma_s}{\text{d}s} &=& {\partial {\cal K}\over \partial \phi}
  + {(3 -4 \mu +\mu^2) \varepsilon \sin\phi\over 2(1+\varepsilon \cos\phi)^2}
  \ {\partial {\cal K}\over \partial \Phi}\cr
  &=&  -\frac{\varepsilon\sin\phi }{(1+\varepsilon\cos\phi)^2}
\bigg(  (1-\mu)\Vert u\Vert^2\bigg(\frac{1}{\Vert\pi(u)+(1,0,0)\Vert}+\pi_1(u)\bigg)\cr
  &+&\mu+\frac 12\Vert u\Vert^2\left(\pi_1^2(u)+\pi_2^2(u)-\pi_3^2(u)\varepsilon\cos \phi\right)+\frac{(1-\mu)^2}{2}\Vert u\Vert^2\cr
&+& {1\over 2}(1+\varepsilon\cos\phi)
\Vert u\Vert^2 \pi_3^2(u)  - {3 -4 \mu +\mu^2 \over 2}\Norm{u}^2\bigg)
\end{eqnarray}
we first replace the only non-polynomial term with the representation
(\ref{Delta4}), 
 and then we simplify the expression obtaining:
\begin{equation}
 \frac{\text{d}\Gamma_s}{\text{d}s} =   -\frac{\varepsilon\sin\phi }{(1+\varepsilon\cos\phi)^2}
\bigg(\mu +\Vert u\Vert^2 (1-\mu)\Delta_4
+{1\over 2}\Vert u\Vert^2\Norm{\pi}^2\bigg)\;.
\end{equation}
Using (\ref{upperbounddelta4}), $\Norm{\pi}^2=\Norm{u}^4$, and $\Norm{u}<\mu^{1\over 6}$ in
the Hill's sphere (as well as $1-\mu < 1$), we easily obtain
(\ref{varGamma}).

\subsection{Proof of inequality (\ref{rho''})}\label{rho''section}

First, we represent $\text{d}^2\rho/\text{d}s^2 = 2 u\cdot\text{d}^2u/\text{d}s^2 + 2 \Norm{\text{d}u/\text{d}s}^2$ 
by replacing $\text{d}^2u/\text{d}s^2$ using the Hamilton equations of ${\cal K}$ given in the
form (\ref{eqn:KfunR}):
\begin{equation*}
  \frac{\text{d}^2\rho}{\text{d}s^2}=\Gamma_s \Norm{u}^2 +  2 \Norm{\frac{\text{d}u}{\text{d}s}}^2
  +{1\over 8} \sum_i (U_i-b_i) \sum_j  u_j \left ( {\partial b_i\over \partial u_j}- {\partial b_j\over \partial u_i}\right )
  -\frac12 \sum_{i} u_{i} {\partial \mathcal{R}_6 \over \partial u_{i}}\;,
  \end{equation*}
and then we replace $2\Norm{\text{d}u/\text{d}s}^2=(1/8)\Norm{U-b}^2$ using ${\cal K}(u(s),\phi(s),U(s),\Phi(s))=0$:
\begin{eqnarray}
  \frac{\text{d}^2\rho}{\text{d}s^2}& = &2\Gamma_s \Norm{u}^2  +{\mu \over {1+\varepsilon \cos\phi}}\cr
& &+{1\over 8} \sum_i (U_i-b_i) \sum_j u_j \left ( {\partial b_i\over \partial u_j}- {\partial b_j\over \partial u_i}\right )
-\frac12\sum_{i} u_{i} {\partial \mathcal{R}_6 \over \partial u_{i}} -\mathcal{R}_6\;.
\label{rho''est}
\end{eqnarray}
Using \eqref{c1}, \eqref{c4} and \eqref{c3}, \eqref{eqn:Uminb}, the terms in the second line of (\ref{rho''est}) are bounded by:
\begin{equation*}
  \norm{\frac12\sum_{i} u_{i} {\partial \mathcal{R}_6 \over \partial u_{i}} + \mathcal{R}_6 }
  \leq \frac12 \norm{\sum_{i} u_{i} {\partial \mathcal{R}_6 \over \partial u_{i}}}+\norm{\mathcal{R}_6}
\leq   \left(c_1+\frac12 c_4\right)\frac{1}{1-\varepsilon} \Norm{u}^6\;,
 \end{equation*} 
$$
  \norm{
    \sum_i (U_i-b_i) \sum_j u_j \left ( {\partial b_i\over \partial u_j}-{\partial b_j\over \partial u_i}\right )} \leq
  \Norm{U-b}\sum_i \norm{\sum_ju_j \left ( {\partial b_i\over \partial u_j}- {\partial b_j\over \partial u_i}\right )}\leq 
  $$
  $$
  \leq c_3 \Norm{U-b} \Norm{u}^3 
  \leq
 c_3 \cU \Norm{u}^3 \sqrt{\Gamma_0\Norm{u}^2 + {\mu \over 1-\varepsilon}}\; .
$$
Therefore, using $\Gamma_s\geq \Gamma_0/2$ and \eqref{usphere}, there exists a constant
 $\gamma_3$ such that:
\begin{eqnarray}
  \frac{\text{d}^2\rho}{\text{d}s^2}& \geq & \Gamma_0 \Norm{u}^2  +{\mu \over 1+\varepsilon}
  - {c_3c_6\over 8} \Norm{u}^3 \sqrt{\Gamma_0\Norm{u}^2 + {\mu \over 1-\varepsilon}}- \frac{2c_1+c_4}{2(1-\varepsilon)} \Norm{u}^6\cr
  & \geq & \Gamma_0 \Norm{u}^2  +{\mu \over {1+\varepsilon}}
  -\gamma_3 \Norm{u}^3\max \left \{  \Norm{u}\sqrt{{ \Gamma_0}}, {\sqrt{\mu} \over
    \sqrt{1-\varepsilon}}, { {\sqrt{\mu} \over 1-\varepsilon} } \right \}\cr
  & = & \Gamma_0 \Norm{u}^2  +{\mu \over {1+\varepsilon}}
  -\gamma_3 \Norm{u}^3\max \left \{  \Norm{u}\sqrt{{ \Gamma_0}}, { \sqrt{\mu} \over 1-\varepsilon} \right \}\;.
\label{rho''est2}
\end{eqnarray}
Let us now consider the two cases:
\begin{itemize}
\item[-] In the case: $\Norm{u}\sqrt{{ \Gamma_0}} \leq \sqrt{\mu}/(1-\varepsilon)$,
  inequality (\ref{rho''est2}) becomes:
  $$
  \frac{\text{d}^2\rho}{\text{d}s^2} \geq  \Gamma_0 \Norm{u}^2  +{\mu \over {1+\varepsilon}}
  -\gamma_3  \Norm{u}^3   {\sqrt{\mu}\over {1-\varepsilon}} \geq
\Gamma_0 \Norm{u}^2  +{\mu \over {1+\varepsilon}}
-\gamma_3  \Norm{u}^2   {\mu \over (1-\varepsilon)^2\sqrt{\Gamma_0}}
$$
$$
\geq
  {\Gamma_0\over 2} \Norm{u}^2  +{\mu \over {1+\varepsilon}}\;,
  $$
  where the last inequality is true if we assume:
  \begin{equation}
  \mu\leq (1-\varepsilon)^2{\Gamma_0^{3\over 2}\over 2 \gamma_3}\;.
  \label{muthr1}
  \end{equation}
\item[-] In the case $ \Norm{u}\sqrt{{ \Gamma_0}} \geq  \sqrt{\mu}/(1-\varepsilon)$,
   inequality (\ref{rho''est2}) becomes:
  $$
  \frac{\text{d}^2\rho}{\text{d}s^2} \geq  \Gamma_0 \Norm{u}^2  +{\mu \over {1++\varepsilon}}
  -\gamma_3  \Norm{u}^4   \sqrt{\Gamma_0}
  $$
  $$
  \geq
  {\Gamma_0\over 2} \Norm{u}^2  +{\mu \over {1+\varepsilon}}\;,
  $$
  where the last inequality is true for all $u$
  satisfying:
  $$
  \Norm{u}^2 \leq { \sqrt{\Gamma_0} \over 2 \gamma_3}\;,
  $$
  which is satisfied for all $ u\in B(\mu^{\frac16})$ if 
  $\mu^{1\over 3}\leq \sqrt{\Gamma_0}/(2 \gamma_3)$, 
  or equivalently:
\begin{equation}
  \mu \leq {1\over 8\gamma_3^3} \Gamma_0^{3\over 2}\;.
  \label{muthr2}
  \end{equation}
\end{itemize}
Therefore, we have proved inequality (\ref{rho''}) 
if $\mu$ satisfies (\ref{muthr1}) and (\ref{muthr2}), which are both satisfied
by inequality (\ref{thrmu}) for any $c_7$ such that:
$$
c_7 \leq \min \left\{ {1\over 8\gamma^3_3},{1\over 2 \gamma_3} \right\}\;.
$$

\section{Numerical examples} 
\label{sec:numericexamples}

\subsection{The advantage of regularization: a numerical test}
\label{sec:numeric}
In order to assess the effectiveness of KS regularization of the ER3BP 
near the singularity at $P_2$ for numerical integrations we
consider a fictitious simple scenario which is nevertheless representative
(for the choice of the initial conditions) of
realistic close encounters in the Solar System, such as the non-coplanar close encounters of comets with Jupiter, in the Sun-Jupiter ER3BP (here identified by the values $\mu=9.536433730801362\cdot10^{-4}$, $\varepsilon=0.0489$). Precisely we consider the case, critical for the numerical integrations, of fast close encounters, where here we mean `fast' both in the sense of Section \ref{sec:fastcloseencounters} (see the caption of Fig. \ref{fig:backforw}) and that the close encounter does not produce a temporary capture through the Lagrangian points $L_1,L_2$. We emphasize that fast close encounters are observed for celestial bodies in the Solar System (see for example \cite{GL15,GL17} where the dynamics of comet 67P Churyumov-Gerasimenko, target of the recent Rosetta mission, is discussed in detail), and are used in Astrodynamics to accelerate spacecrafts.\\
\indent In this subsection we consider a model example numerically integrated using a quadruple precision floating-point format with an explicit fixed step numerical integrator of the Runge-Kutta family (Luther's method \cite{luther1968explicit}, RK6 for brevity) to analyze possible gains in the use of the KS regularization. The choice of using explicitly a fixed step integrator is motivated by the need to avoid any  interference of a variable step strategy with the regularization, which automatically performs the reduction of the step size by adopting a fictitiuos independent variable. We also remark that, even if the RK6 integrator is not symplectic, it does not produce a relevant energy loss since fast close encounters occurr in small time intervals.\\
\indent Therefore, we choose orbits with initial conditions characterized by 
a high initial energy $\mathscr{E}> \mathscr{E}_4\coloneqq J(x_{L_4},y_{L_4},z_{L_4},x^{\prime}_{L_4},y^{\prime}_{L_4},z^{\prime}_{L_4},f_0)$, where $J$ is the $f$-dependent Jacobi ``integral''\footnote{We recall that $J$ is not a first integral for the eccentricities $\varepsilon>0$; therefore, the choice $\mathscr{E}>\mathscr{E}_4$ provides only a trial initial condition for having a fast close encounter.}:
\begin{multline}
\label{eqn:jacobi}
J(x,y,z,x^{\prime},y^{\prime},z^{\prime},f)=\frac12\left((x^{\prime})^2+(y^{\prime})^2+(z^{\prime})^2\right)\\
-\frac{1}{1+\varepsilon\cos f}\left(\frac{1-\mu}{d_1}+\frac{\mu}{d_2}+\frac{1}{2}(x^2+y^2-z^2\varepsilon\cos f)\right)\;,
\end{multline}
having, in the vicinity of the secondary body, an important deflection of the trajectory with respect to the
solutions of the Kepler problem defined by the Sun. Moreover,
we consider orbits which are non-planar with Jupiter's orbit. An efficient
way to visualize the outcome of the regularization on fast close encounters
is to consider the initial conditions already at their minimum distance
from Jupiter, with $\mathscr{E}> \mathscr{E}_4$, with inclination
different from $0$, then to numerically integrate the orbit by first running a backward integration up to a sufficiently large distance (at least $d_2=\|P-P_2\|>1=\|P_2-P_1\|$), and then to switch to a forward integration lasting exactly twice the number of iterations of the previous operation, so that the upshot produces almost equal-length branches before and after the encounter (blue and red lines in top panels of Fig. \ref{fig:backforw}). In such a way, we appreciate the whole dynamics with the deviation caused by the planet. More details on the choice of the initial conditions are given in the caption of Fig. \ref{fig:backforw}; the physical parameters are drawn from the NASA Planetary Fact Sheet.\footnote{\url{https://nssdc.gsfc.nasa.gov/planetary/factsheet/}}
The details of implementation
  of the numerical integration, including also the 
  formulas for the computation of the initial condition in $(u,\phi,U,\Phi)$ variables for any given Cartesian
  initial condition $(x,y,z,p_1,p_2,p_3)$, are provided in Appendix  \ref{subsecapp:implem}.

\begin{figure}
	\centering
	\includegraphics[scale=0.43]{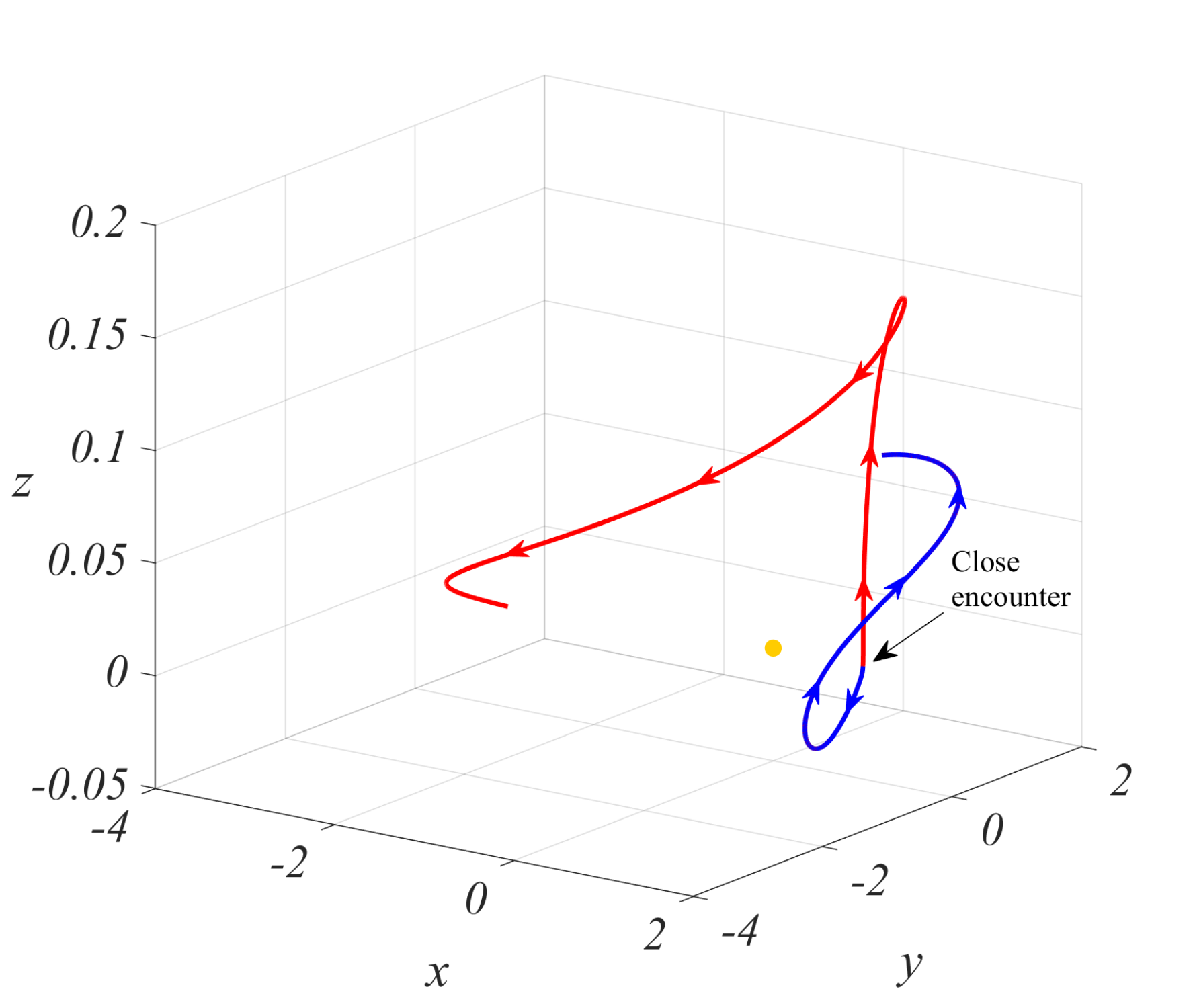}\\
	\includegraphics[scale=0.39,trim={11cm 0.9cm 11cm 1.4cm},clip]{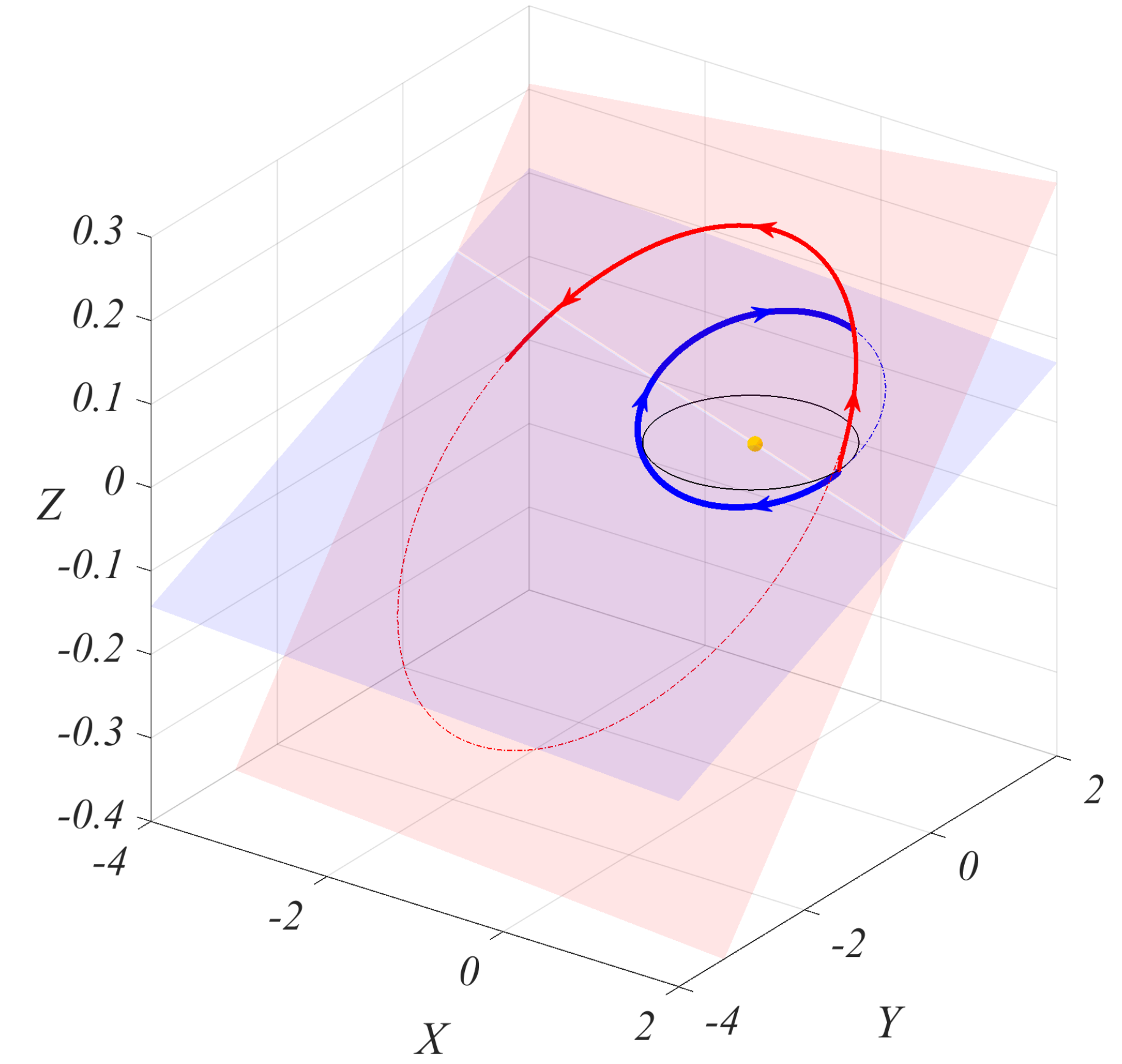}
	\includegraphics[scale=0.39,trim={11cm 0.9cm 11cm 1.4cm},clip]{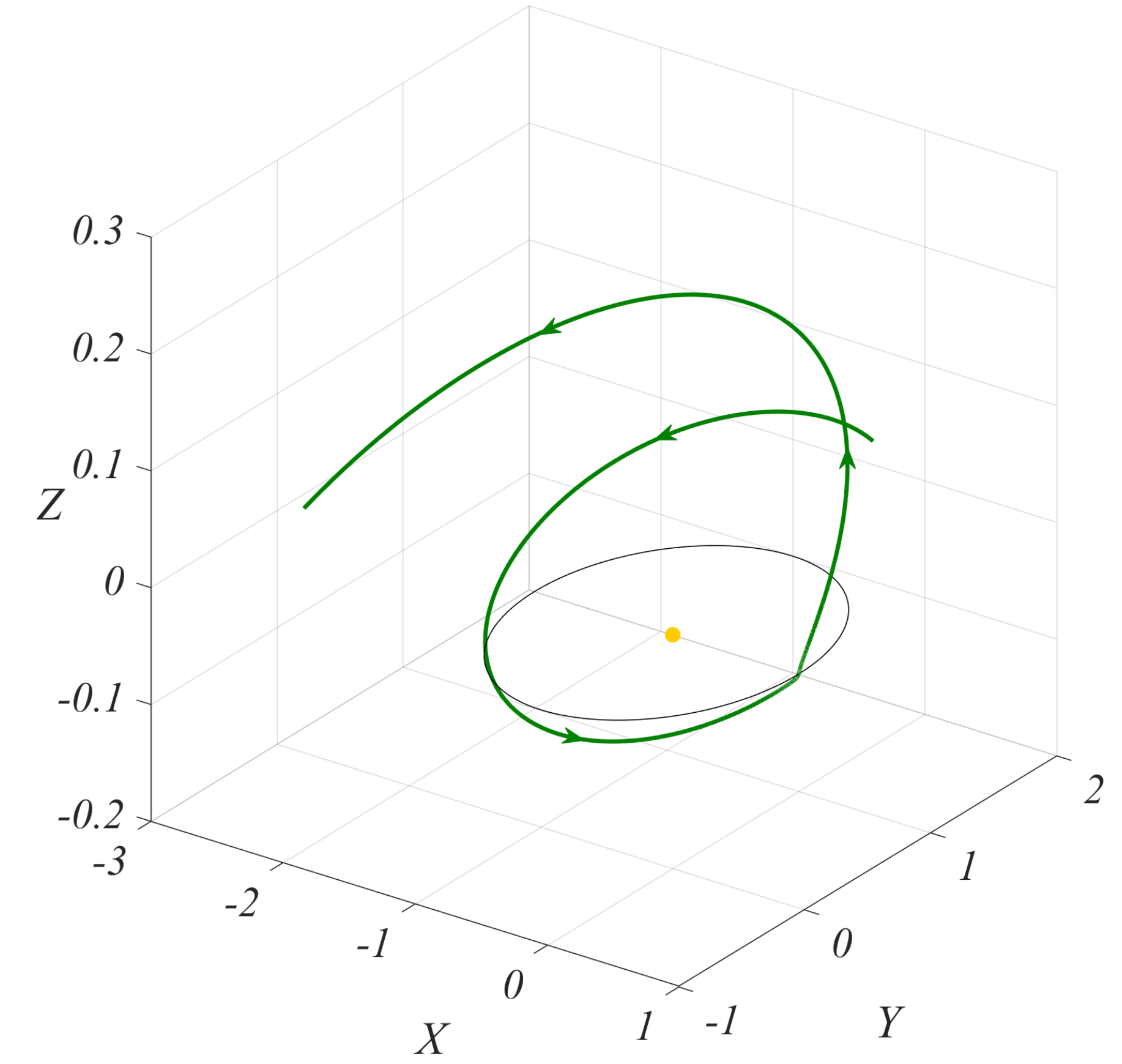}
	\caption{Physical orbit (reported in convenient aspect ratio for visual clarity) integrated backward in true anomaly up to $f=-2\pi$ and then forward up to $f=2\pi$ following the arrow heads for $x_0=1-\mu+ 1.921451079855507\cdot10^{-3}$ ($\approx 0.01$ AU of altitude), $y_0=z_0=f_0=0$, $p_{x,0}=0.2$, $p_{y,0}=1.8$, $p_{z,0}=0.6$. The corresponding KS-transformed quadruple precision initial data are:
	\begin{align*}
	u_1&=0.0438343595807618585658005372351908591\;,\\
	U_1&=0.0175337438323047538346610707549189101\;,\\
	U_2&=0.0702185800222737827036567637151165400\;,\\
	U_3&=0.0526012314969142580345362603111425415\;,\\
	\Phi&=-1.38220656687993415599045818608844111\;,\\
	u_2&=u_3=u_4=\phi=U_4=0\;.
    \end{align*}
	The yellow dot symbolizes the Sun, whereas the black thin style curve represents Jupiter's elliptic motion. \textbf{Top panel}: Cartesian version (backward in blue overlapping the forward in red) traced in the rotating-pulsating frame $Oxyz$. \textbf{Left bottom panel}: Cartesian backward (blue) and forward (red) trajectory in the inertial barycentric frame $OXYZ$ (Appendix \ref{subsecapp:xyztoXYZ}) with osculating heliocentric ellipses belonging to mutually inclined planes. \textbf{Right bottom panel}: KS integration of the inertial trajectory in the forward case. The close encounter is of hyperbolic type (fast): $\Gamma_0=1.4282186\le\Gamma_s\le\Gamma_s\vert_{\min d_2}<(3/2)\Gamma_0$ in $B(\mu^{\frac13})$.}
	\label{fig:backforw}
\end{figure}
\begin{figure}
	\centering
	\includegraphics[scale=0.34,trim={8.5cm 0cm 8.5cm 0cm},clip]{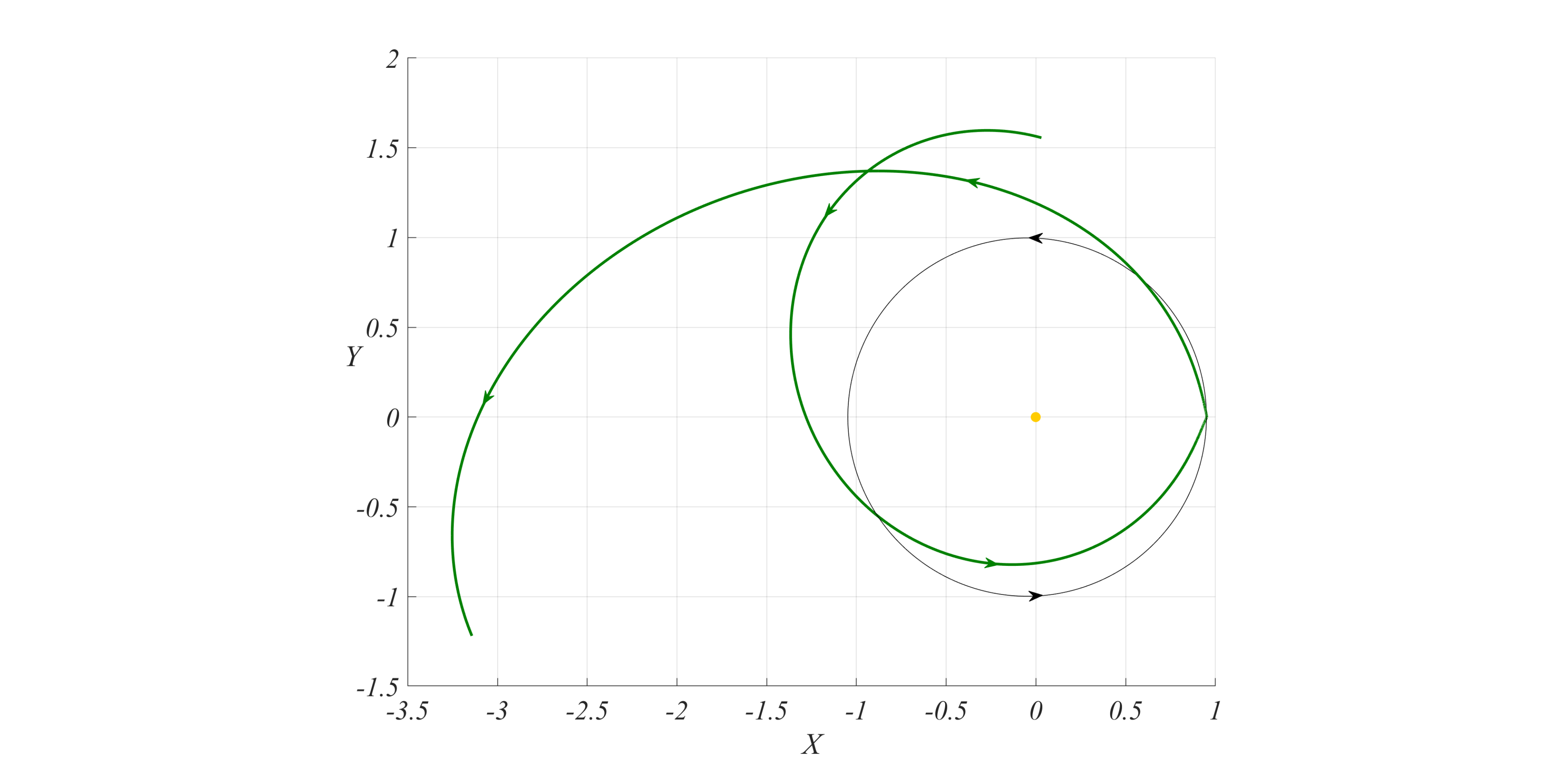}\\
	\includegraphics[scale=0.34,trim={3cm 5cm 3cm 5cm},clip]{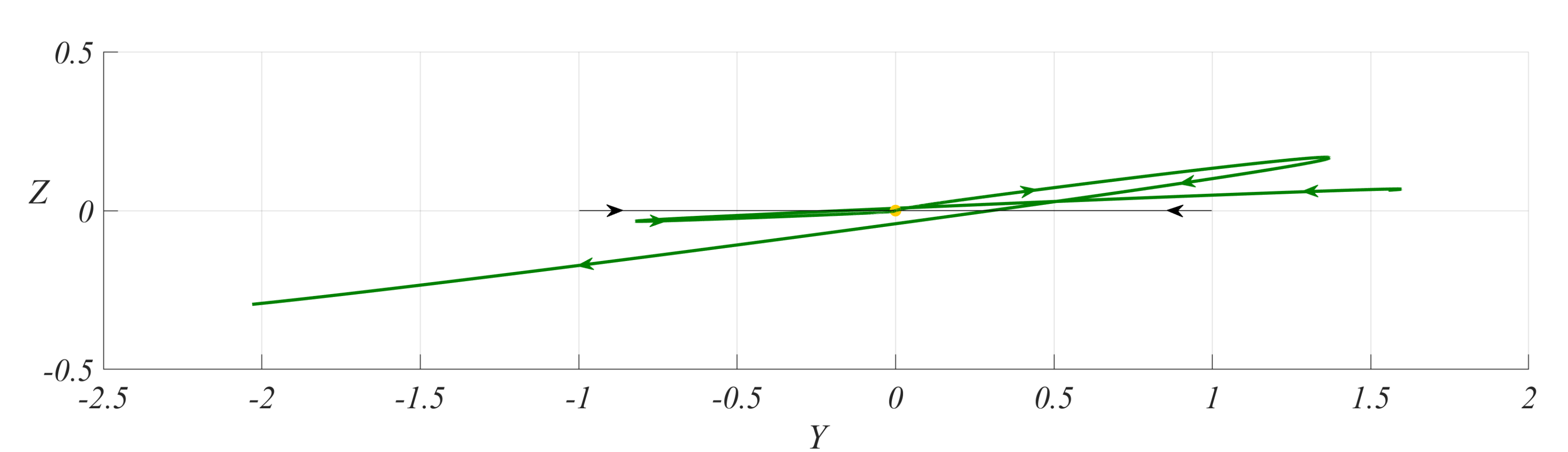}\\
	\includegraphics[scale=0.34,trim={3cm 5cm 3cm 5cm},clip]{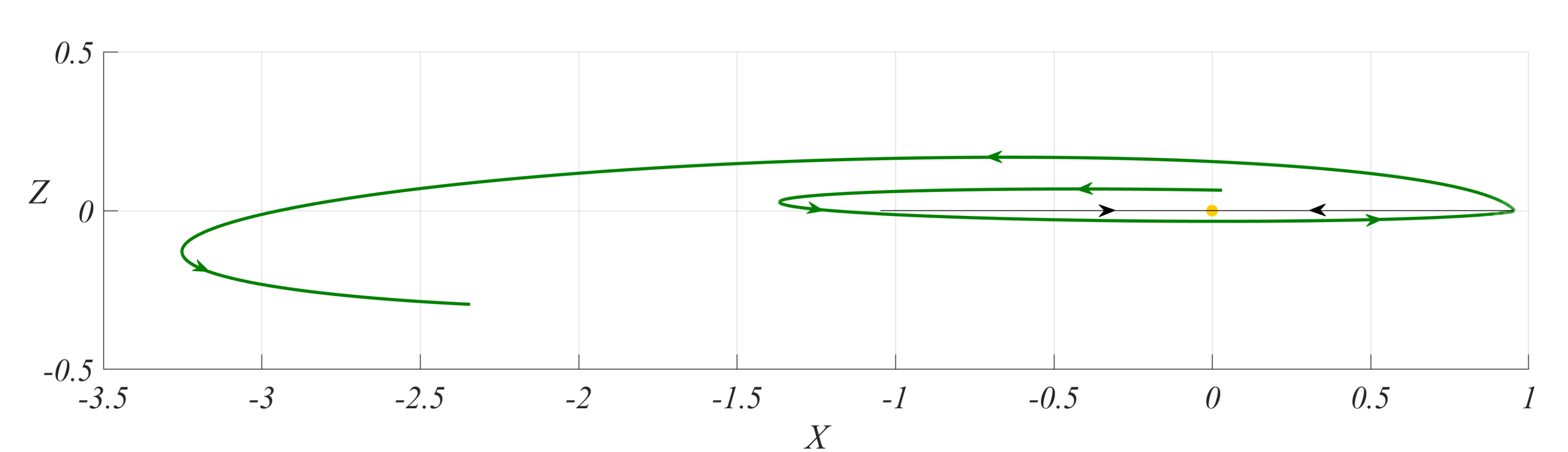}
	\caption{Projections on the coordinate planes of the right bottom panel of Fig. \ref{fig:backforw} on equal axis aspect ratio.}
	\label{fig:orbproj}
\end{figure}

In Figs. \ref{fig:backforw} and \ref{fig:orbproj}  we illustrate the particle's orbit: we set initial values so that the planetary flyby is of hyperbolic gravity-assist type and the heliocentric paths before and after the encounter are almost Keplerian ellipses.\\
\indent In Table \ref{table:KSCart} we compare the numerical integration of the close encounter represented in Figs. \ref{fig:backforw} and \ref{fig:orbproj} using both the Cartesian and the KS regularized equations of motion on a smaller true anomaly range, for different values of the integration steps. Precisely, using as initial condition the point of closest approach, we
integrate the equations of motion backward up to $f_{-}\approx-0.5$, and then forward up to $f_{+}\approx 0.5$, since
$f_{+}-f_{-}\approx1$ represents approximately the interval of the true anomaly in which the
close encounter takes place. We compare both the conservation
of the extended Hamiltonians $\widehat{\mathcal{H}}=\mathcal{H}+\Phi$, $\widehat{\mathscr{K}}$ and the value of $\|r\|$, $r=(x,y,z)$, at the end of the two numerical integrations. Since the numerical integration of the  Cartesian equations of motion is carried out with a fixed step $\Delta f$, while the numerical integration of the regularized equations of motion is carried out with a fixed step $\Delta s$, in order to compare the outputs of the two numerical integrations at the same values of $f$ we add a further iteration to the Cartesian numerical
integration to reach the final value of $f$ obtained with
the KS numerical integration. 
We note the sharp advantage of the KS algorithm in terms of computational efficiency, providing a sharp reduction of the number of iterations needed to maintain a high level of accuracy in the representation of the trajectory around $P_2$.\\
We finally argue that, to maximize the computational efficiency in a long simulation, one should implement as usual a switching tool which works with the $\mathcal{H}$-equations out of a ball centered at $P_2$, say the body's Hill's sphere, and the $\mathcal{K}$-equations once inside.\\
\indent As conclusive remark, since $\mathcal{K}$ and $l$ are smooth functions in a neighborhood of $P_2$, we get that the invariances given by $\mathcal{K}(u,\phi,U,\Phi)=0$, $l(u,U)=0$ are numerically satisfied with high precision, and eventually the machine precision is quickly reached after lowering $\Delta s$ by about a factor of a hundredth (in a convergence profile, the rate of decay reflects the accuracy of the method, i.e., $\mathcal{O}(\Delta s^6)$).

\begin{table}[H]
	\[
	\begin{array}{lccc}
	\toprule
	\Delta f \text{ \textit{(Cart.)}} & \|r(-0.5066821124431412)\| & \|r(0.4961307051398083)\| & \textit{\# iter.} \\
	\midrule
	2\pi\cdot10^{-6} & 0.8553075048550535 & 0.9760051057296899 & 240244\\
	2\pi\cdot 10^{-5} & 0.8553075048542582 & 0.9760051057288172 & 24026\\
	2\pi\cdot 10^{-4} &  0.8553060796173549 & 0.9760054080001320 & 2404 \\
	2\pi\cdot 10^{-3}\;(*) & 0.8248588821498852 &
	0.9897100124542644 & 241\\
	\midrule
	\Delta s \text{ \textit{(KS reg.)}} & \|r(f(-3.7\pi))\| & \|r(f(3.5\pi))\| & \textit{\# iter.}\\
	\midrule
	\pi\cdot10^{-4} & 0.8553075048550521 & 0.9760051057296942 & 109000 \\
	\pi\cdot10^{-3} & 0.8553075048550521 & 0.9760051057296942 & 10900\\
	\pi\cdot 10^{-2} & 0.8553075048550522 & 0.9760051057296968 & 1090\\
	\pi\cdot 10^{-1} & 
	0.8553075050607468 & 0.9760051591505222 & 109\\
	\bottomrule
	\end{array}
	\]
	\[
	\begin{array}{lcc}
	\toprule
	\Delta f \text{ \textit{(Cart.)}} & |\widehat{\mathcal{H}}(-0.5066821124431412)| &
	  |\widehat{\mathcal{H}}(0.4961307051398083)|\\
	\midrule
	 2\pi\cdot10^{-6} & 
	 1.0417562295\cdot 10^{-18} & 1.0277827090\cdot 10^{-18}\\
	 2\pi\cdot 10^{-5} & 9.3757489321\cdot 10^{-13} & 7.9843639352\cdot 10^{-13}\\
 	2\pi\cdot 10^{-4} & 1.1893484533\cdot 10^{-7} & 8.5748939646\cdot 10^{-7}\\
	 2\pi\cdot10^{-3}\;(*) & 
	 8.0281428133\cdot 10^{-2} & 0.10590853333\\
	\midrule
	\Delta s \text{ \textit{(KS reg.)}} & |\widehat{\mathscr{K}}\,(-3.7\pi)| & |\widehat{\mathscr{K}}\,(3.5\pi)|\\
	\midrule
	 \pi\cdot 10^{-4} & 1.3746151644\cdot 10^{-27} &
	 1.3290033656\cdot 10^{-28}\\
	 \pi\cdot10^{-3} & 
	 1.3738069068\cdot 10^{-21} & 1.3119148531\cdot 10^{-22}\\
	 \pi\cdot 10^{-2} & 1.3654070424\cdot 10^{-15} &
	 1.1227698042\cdot 10^{-16}\\
	 \pi\cdot10^{-1} & 1.2545211218\cdot 10^{-9} & 
	 3.0569361253\cdot 10^{-10}\\
	\bottomrule
	\end{array}
	\]
	\caption{Comparison between Cartesian and KS integration for four consecutively increased step sizes in a small neighborhood of $f=s=0$ (close approach). The propagations are performed backward in time up to 
		\[
		f(-3.7\pi)=-0.506682112443141208003735413674982089\;,
		\]
	then forward up to
	\[
	f(3.5\pi)=0.496130705139808336532715403656106249\;,\]
	according to choices of $s$ such that the true anomaly interval is almost symmetric with respect to the origin and sample nodes are multiple integers of every $\Delta s$ considered. For each case the step size is eventually adapted in order to evaluate the singular and regularized solution at same corresponding times (see text). The asterisk in the last Cartesian experiment indicates failure of the adopted numerical method (explicit RK6) to compute the orbit accurately. \textbf{Top panel}: norm of the solutions (rounded to 16 significant digits) and total number of iterations. \textbf{Bottom panel}: degree of conservation of the singular extended Hamiltonians (rounded to 11 significant digits), where $\widehat{\mathcal{H}}(0)=\widehat{\mathscr{K}}\,(0)=0$.}
	\label{table:KSCart}
\end{table}

\subsection{Detection of multiple close encounters with RFLIs}
\label{sec:RFLI}
Regularized fast Lyapunov indicators, introduced in
\cite{CLSF,LGF11,GL2013} (see also \cite{GL23}), are defined for the ER3BP
by:
\begin{equation}
\text{RFLI}(\xi_0,w_0;f_0,F) = { \max_{ f_0\leq f \leq F}}
\log_{10} {\Norm{w(s(f))}\over \Norm{w_0}}\;,
\label{flireg}
\end{equation}
where:
\begin{itemize}
\item $\xi_0=(r_0,r_0^{\prime})$ is the initial position and initial velocity
  vector given in Cartesian coordinates; $f$ is the true anomaly of the elliptic
  motion and $f_0$, $F$ its initial and final values;  $s(f)$ denotes the proper time expressed
as a function of $f$;

\item $w(s)$ is the solution of the variational equations
  of Hamiltonian ${\cal K}$ obtained from the variational matrix \eqref{eqn:matX} with initial conditions $w(s(f_0))=w_0$, 
  computed for an orbit with initial
conditions $u(f_0)$, $\phi(f_0)$, $U(f_0)$, $\Phi(f_0)$ provided by a local inversion
of the KS transformation (cf. \eqref{eqn:invpibarrel}--\eqref {eqn:invpi} in Appendix \ref{subsecapp:implem}; by changing the local inversion map, a
change in the initial tangent vector should be applied accordingly, see \cite{GL18}; however, in
the experiments below we do not need to apply any local inversion map for the choice of the initial tangent vector and it suffices to fix the same vector in $\mathbb{R}^8$ for all the integrations, see \cite{GL18} for more details).

\end{itemize}
According to \cite{GL14,LG16,GL17,GL18}, we also consider the modified indicator mFLI
which is suitable to detect the fast close encounters of the ER3BP:
\begin{equation}
{\rm mFLI}_\chi(\xi_0,w_0;f_0,F) = \max_{ f_0\leq f \leq F} \bigintsss_0^{s(f)} 
{\chi(q(f(s))){w(s)\cdot\displaystyle\frac{\text{d}w}{\text{d}s}(s)\over \Norm{w(s)}^2}}
\text{d}s\;,
\label{modflireg}
\end{equation}
where:
\begin{itemize}

\item $q(f)=r(f)-(1-\mu,0,0)$ represents the solution
of the equations of motion with initial conditions $\xi_0$; 

\item   $f(s)$ is the true anomaly of the elliptic motion expressed
as a function of $s$;

\item $\chi(q)$ is a function depending on a parameter $\lambda>0$:
\begin{equation}
\chi(q)=
\begin{cases}
\displaystyle \qquad\qquad\qquad1 & \displaystyle\text{if }\Norm{q}\leq {\lambda \over 2}\;, \\
\displaystyle {1\over 2}\left[\cos\left(\left({\Norm{q}  \over \lambda}
-{1\over 2}\right)\pi\right)+1\right]  &  \displaystyle\text{if } {\lambda \over 2} <  \Norm{q} 
\leq  {3\over 2}\lambda\;,\\
\displaystyle \qquad\qquad\qquad 0  & \displaystyle\text{if } \Norm{q} > {3 \over 2}\lambda\;,
\end{cases}
\label{chiq}
\end{equation}
 where $\Norm{q}=d_2$ is the Cartesian distance $\Norm{P-P_2}$ and $\lambda$ is a parameter, that for close encounters is conveniently set as from
  $1$ to $2$ Hill's radii $r_h$ of $P_2$.\\
\end{itemize}

In the following experiments, we aim to detect multiple fast close encounters for small values of $\mu<1/10$ using the indicators RFLI and mFLI above. We consider the Sun-Earth spatial ER3BP ($\mu=3.00347\cdot 10^{-6}$, $\varepsilon=0.0167$).\\	
\indent We start from a planar reference orbit characterized by a remarkably fast close encounter with the Earth, i.e., with $\Gamma_0>0$ sufficiently large entering the Hill's sphere; then we suitably vary initial conditions $(x,x^{\prime})$ in a neighborhood of those of the reference orbit to explore the nearby phase space. We take initial conditions corresponding to a $2$:$3$ mean-motion resonant orbit in the Keplerian approximation intersecting the Earth's orbit, with pericenter interior to the Earth's trajectory and aligned with the planet's one on the $X$-axis of the inertial frame. We synchronize times in order to have an exact collision in the future at the first intersection on the Keplerian ellipses by starting with $P$ and $P_2$ located well apart, say, for convenience, at a relative distance $\Norm{P-P_2}(f_0)=1/\sqrt{2}>\Norm{P_2-P_1}(0)/2=(1-\varepsilon)/2$ expressed in the inertial frame. The resulting orbit is integrated in Fig. \ref{fig:reforb}, complemented by $\Gamma_s$ during the Hill's sphere crossing.
\begin{figure}
	\centering
	\includegraphics[scale=0.09]{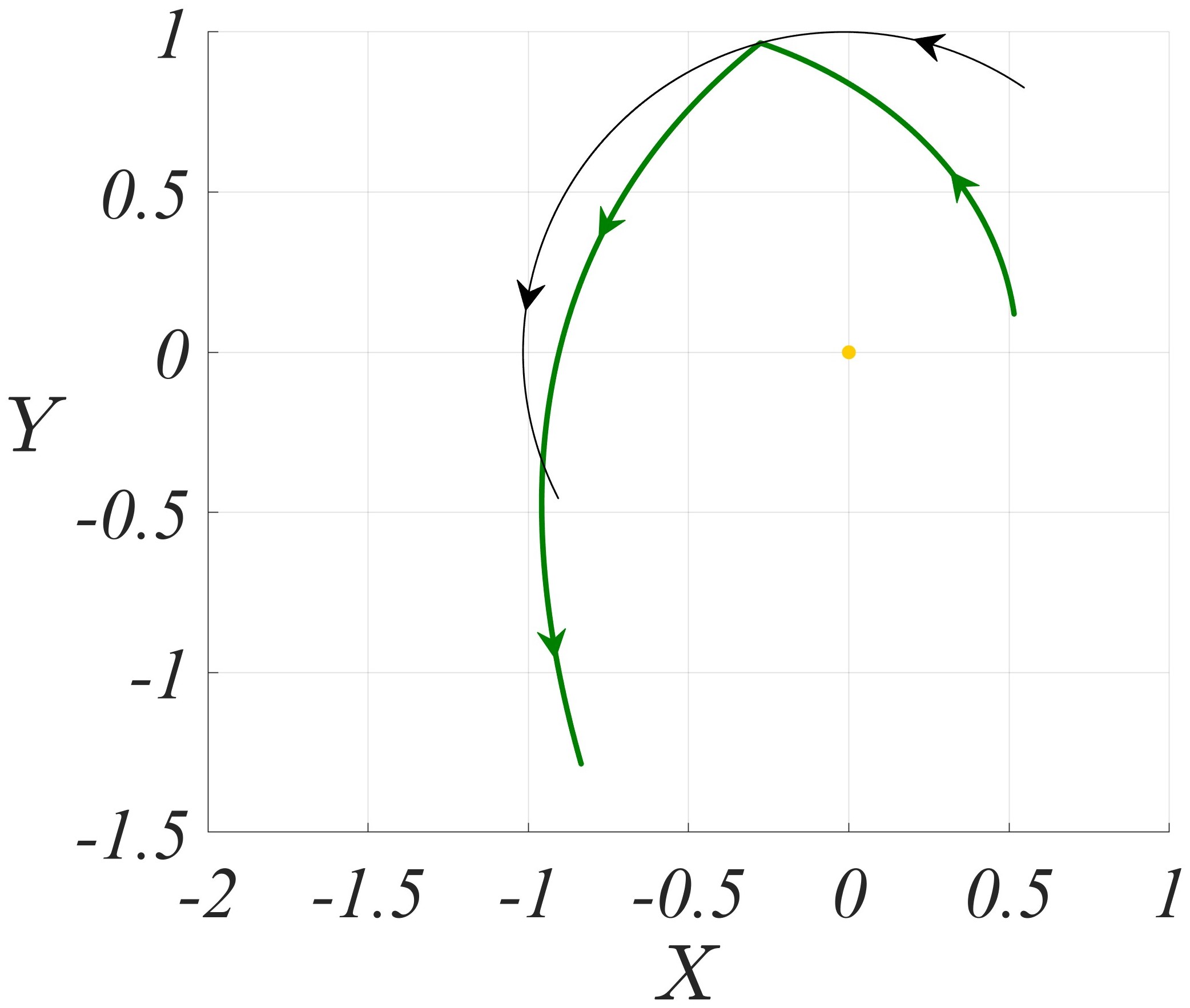}\hspace{1cm}
	\includegraphics[scale=0.09]{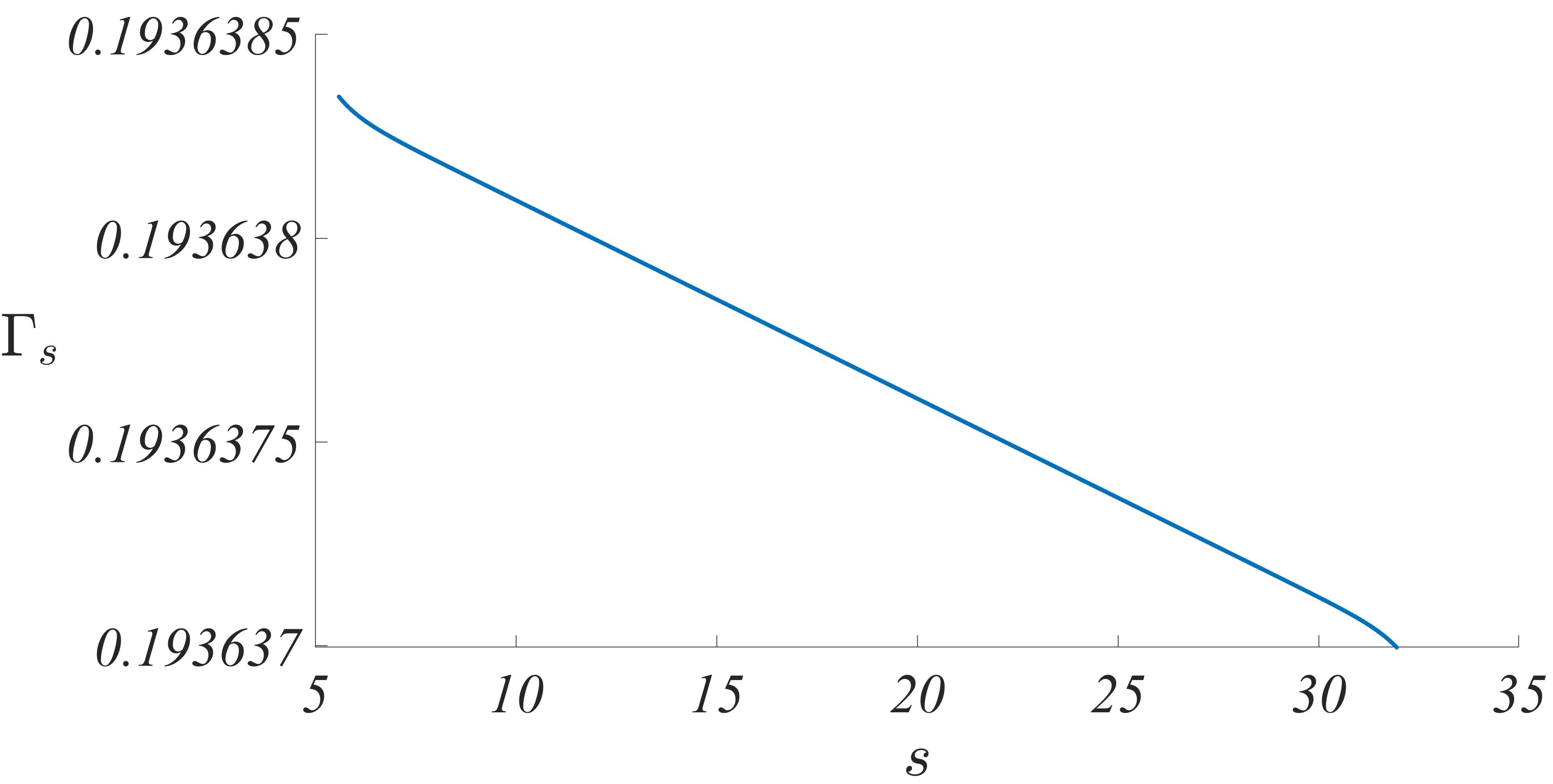}
	\caption{Reference orbit for mFLI/RFLI computations in the Sun-Earth system. \textbf{Left panel}: Cartesian orbit in the inertial barycentric frame $OXYZ$ (green) approaching the Earth (black thin) at $f_c\approx106^{\circ}$ from $f_0=0.9862623425908257$. The initial orbital elements of the test particle are: $a_P(f_0)=1.3103706971044482$ (semi-major axis), $e_P(f_0)=0.6$ (eccentricity), $f_P(f_0)=0.22823102675215523$ (true anomaly), $i_P(f_0)=\omega_P(f_0)=\Omega_P(f_0)=0$ (inclination, argument of pericenter, longitude of the ascending node). The yellow dot represents the Sun. \textbf{Right panel}: trend of $\Gamma_s>0$ throughout the fast close encounter (note, indeed, that $\Gamma_s\approx\Gamma_0\in[\Gamma_0/2,3\Gamma_0/2]$).}
	\label{fig:reforb}
\end{figure}

At this stage, we compute the indicator \eqref{modflireg} on a grid of initial conditions $(x,x')$ equally spaced in the two coordinates, while keeping fixed $y$, $z$, $y'$, $z'$ and setting $\Phi=-\mathcal{H}(x,y,z,p_1,p_2,p_3,f_0)$ to perform the regularization. We slightly move away from the reference orbit until the mFLI portrays fast close encounter structures in the phase space and we therein increase the resolution. Additionally, for a better visualization and to include multiple close encounter orbits, we extend the integration interval from $f_0$ up to $F=f_c+15\pi+2$, with $f_0$, $f_c$ given in the caption of Fig. \ref{fig:reforb}. The outcome reproducing fast close encounter loci is reported in the top panel of Fig. \ref{fig:mFLIEarth}.
\begin{figure}
	\centering
	\includegraphics[scale=0.46]{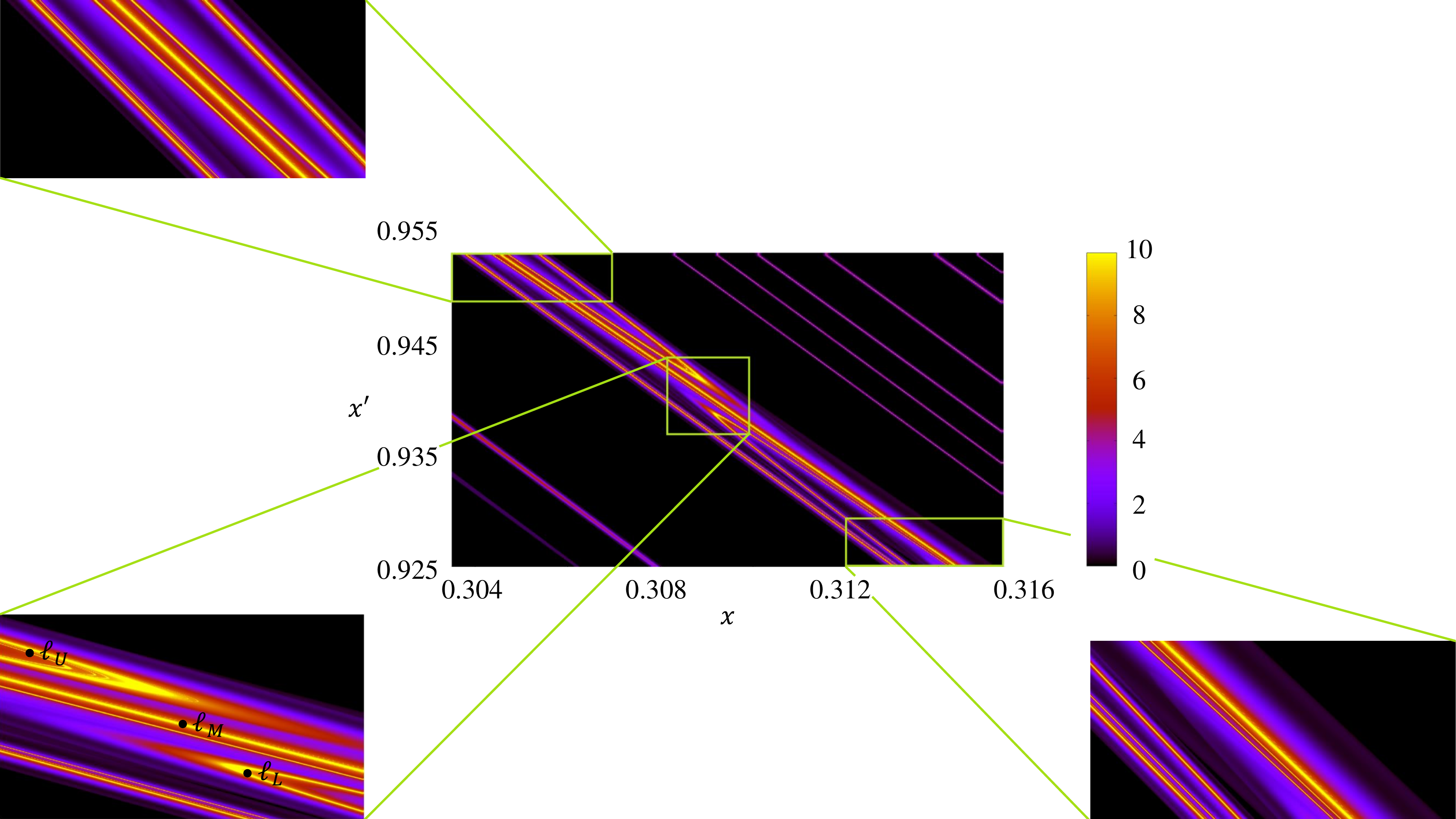}\\
	\vspace{5mm}
	\includegraphics[scale=0.46]{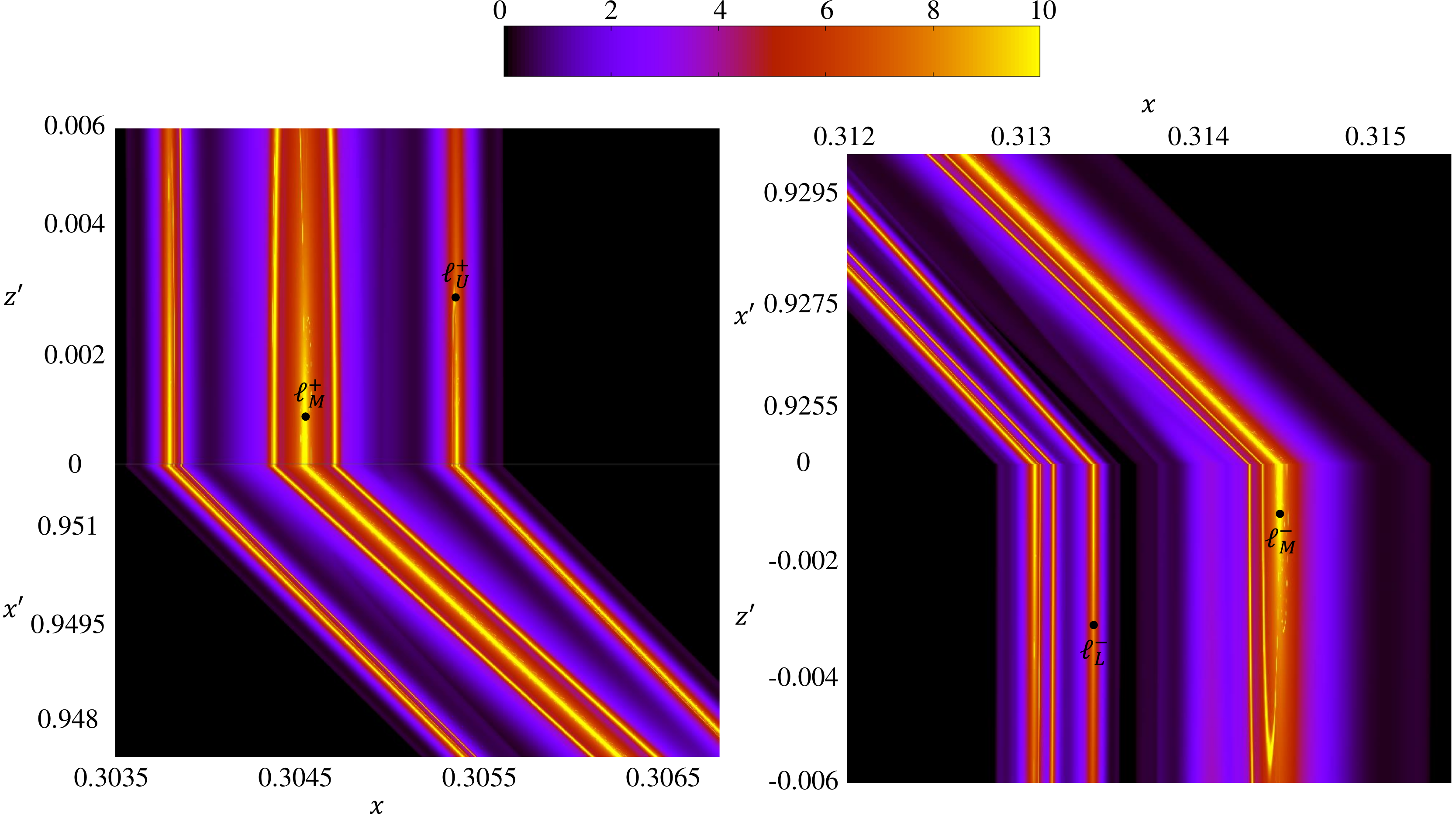}
	\caption{mFLI$_{\chi}$ charts of the Sun-Earth ER3BP over $1000\times1000$ regularly spaced initial conditions close to those of Fig. \ref{fig:reforb}. Here $\lambda=r_H$ in $\chi$ and $w_0\in\mathbb{S}^7$. \textbf{Top panel}: planar section with three magnifications of the close encounter lobes emerging diagonally in the figure. In the central magnification we discern three main lobes (middle, lower, upper), on which sample orbits $\ell_M$, $\ell_L$, $\ell_U$ are taken, respectively. \textbf{Bottom panel}: spatial sections merged to corresponding above planar magnifications of the main lobes. The effect is a continuation of the close encounter structures along the $z'$-axis. Sample orbits $\ell_M^+$, $\ell_U^+$ are taken on the positive $z'$ extensions of the main lobes (left), while sample orbits $\ell_M^-$, $\ell_L^-$ are taken on the negative $z'$ extensions of the main lobes (right).}
	\label{fig:mFLIEarth}
\end{figure}

The spatial continuation of these loci can be inspected by integrating orbits with nonzero initial $z'$ (or $z$) starting from the planar structures. In the bottom panel of Fig. \ref{fig:mFLIEarth} we compute such continuations in the $(x,z')$ section for positive and negative $z'$ upon setting, respectively, $x'=x'_{\text{max}}$ and $x'=x'_{\text{min}}$ for all the grid points, where $x'\in [x'_{\text{min}},x'_{\text{max}}]$ in the planar chart.\\
\indent The lobes depicted in Fig. \ref{fig:mFLIEarth} correspond to single or multiple close encounter orbits, as well as to deeper or less deep encounters. We explore the nature of these lobes by considering the sample orbits (indicated in the same figure) $\ell_M$, $\ell_L$, $\ell_U$ when $z'=0$, $\ell_M^+$, $\ell_U^+$ when $z'>0$ and $\ell_M^-$, $\ell_L^-$ when $z'<0$. We plot in Fig. \ref{fig:sampleorbs} RFLI$(s)$, mFLI$(s)$, $\Gamma_s$ for $s\in[0,s(F)]$ and $d_2(s)$ in a neighborhood containing the Hill's sphere for each selected orbit.
\begin{figure}
	\centering
	\includegraphics[scale=0.1672]{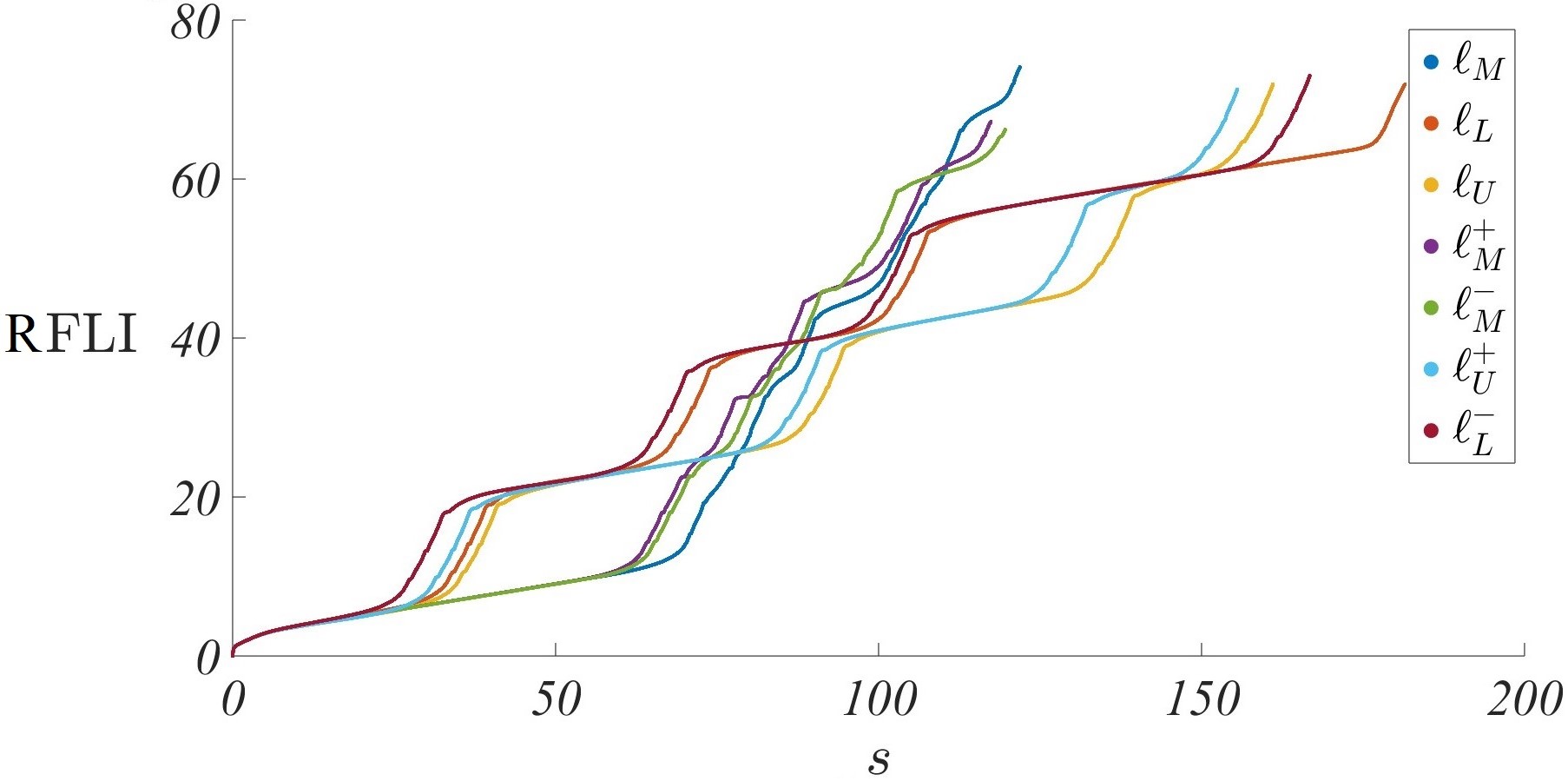}
	\includegraphics[scale=0.1672]{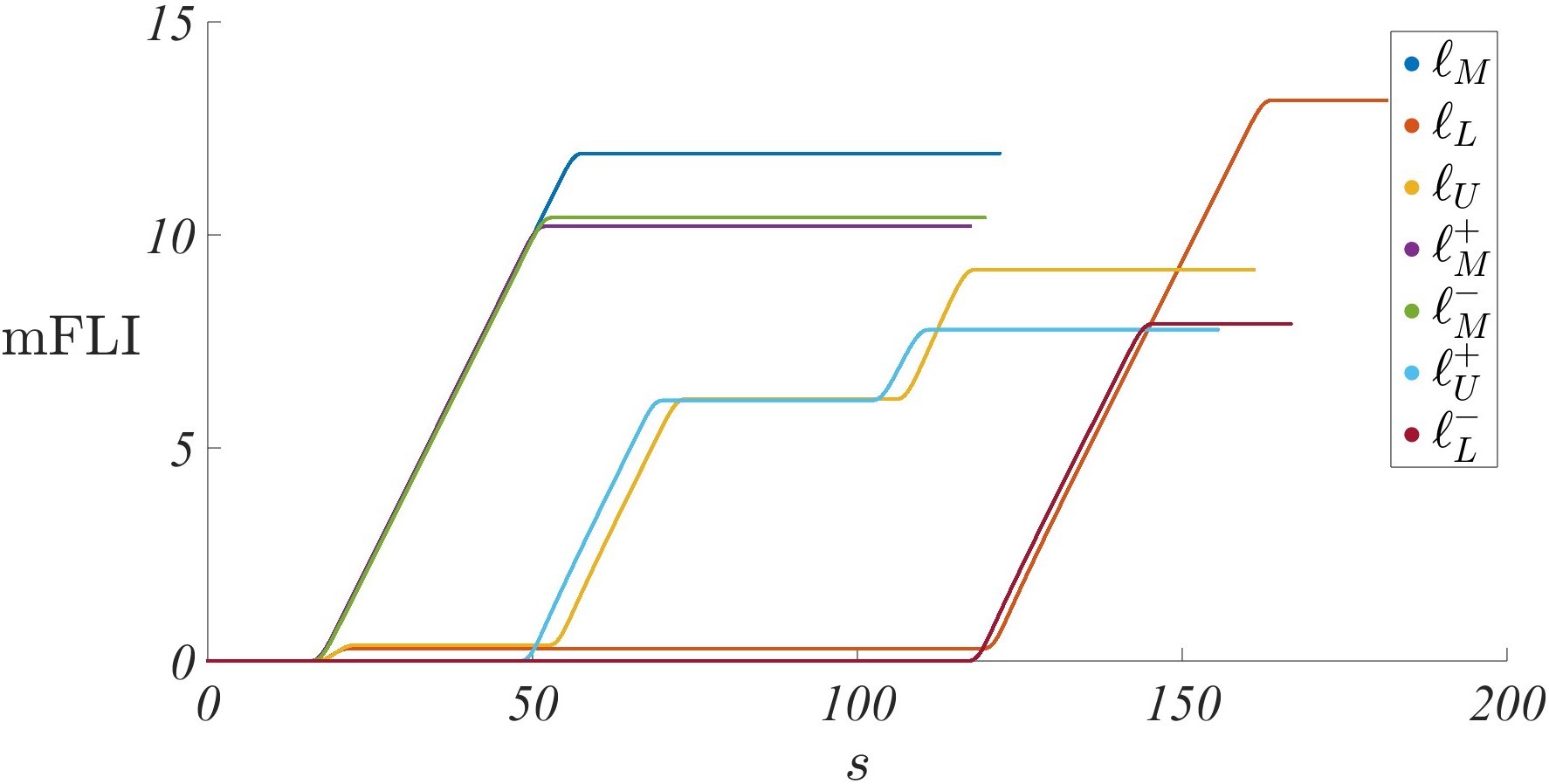}\\
	\vspace{4mm}
	\hspace{2.5mm}\includegraphics[scale=0.1672]{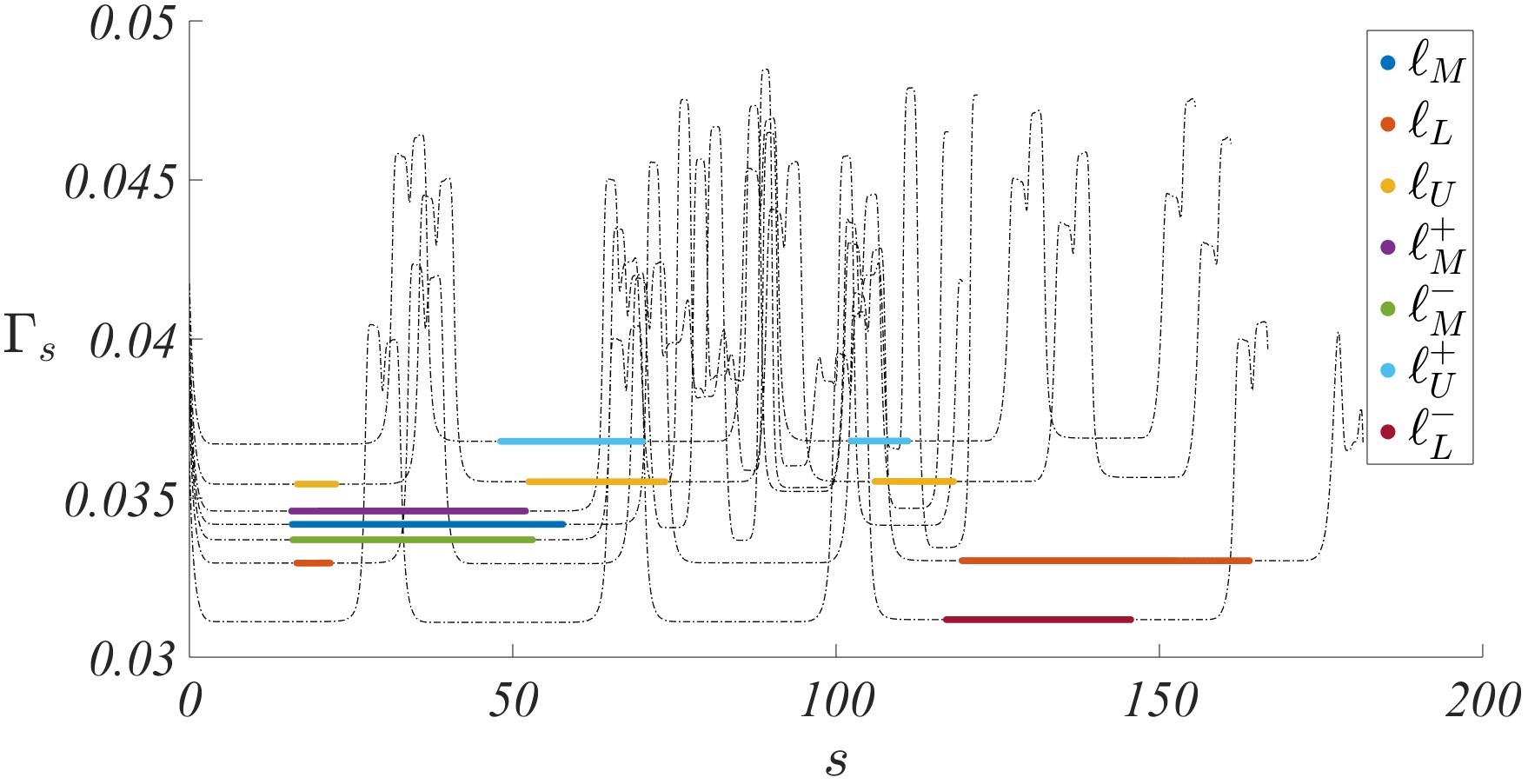}\hspace{1mm}
	\includegraphics[scale=0.1672]{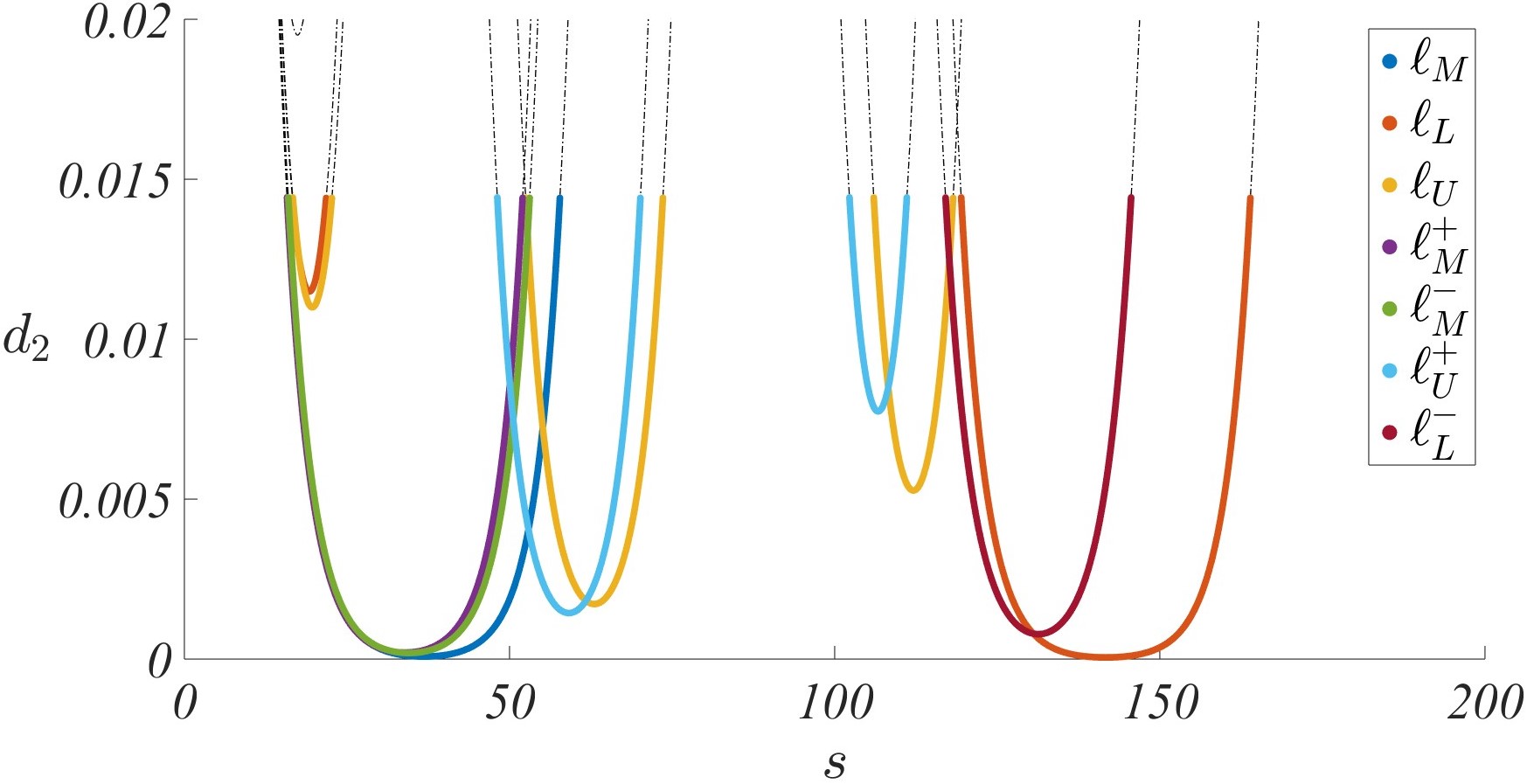}
	\caption{RFLI, mFLI, $\Gamma_s$ and $d_2$ as functions of $s$ for the sample orbits of Fig. \ref{fig:mFLIEarth}. \textbf{Top left panel}: RFLI$(s)$. \textbf{Top right panel}: mFLI$(s)$. \textbf{Bottom left panel}: $\Gamma_s$ along the whole propagation (dashed) and in $B(\mu^{\frac13})$ (solid colored). \textbf{Bottom right panel}: $d_2(s)$ in $B(\mu^{\frac13})$ (solid colored) and graph continuations outside $B(\mu^{\frac13})$ (dashed).}
	\label{fig:sampleorbs}	
\end{figure}
\\
In the top left plot, we observe the strong divergence of orbits with close initial conditions for all the orbits pinpointed by a growing RFLI.\\
In the top right plot, we can distinguish between orbits facing either 1, 2 or 3 close encounters with the Earth by looking at their jumps in the mFLI. Qualitatively, we can better determine the strength of the approaches by comparing their minimum distance $d_2$ in the bottom right plot: we find that the orbits on the middle lobe (label `$M$') have a single strong close encounter (small $\min d_2$ and high mFLI jump); the orbits on the lower lobe (label `$L$') have a double close encounter on the plane, the first weak (high $\min d_2$ and low mFLI jump) and the second strong, whilst a single medium-strong close encounter going into the space; the orbits on the upper lobe (label `$U$') have a triple close encounter on the plane, the first weak, the second and third medium, whilst a double close encounter going into the space, the first medium and the second weak. According to this analysis, the development of the dynamics in the third spatial dimension appears to reduce the chance for fast close encounters, which seems intuitive given the extra degree of freedom coming into play. Nonetheless, it is worth mentioning that an intrinsic limitation is due to the finite resolution of the portraits in Fig. \ref{fig:mFLIEarth}: more refined grids reveal further distinct branches inside the visible lobes with possible different dynamical features in terms of close encounter detection.\\
Finally, in the bottom plot of Fig. \ref{fig:sampleorbs}, we verify the hyperbolicity of each encounter by looking at its variation of $\Gamma_s$ in the Hill's sphere (almost null, see Fig. \ref{fig:reforb}), which fulfills the properties of Section \ref{sec:fastcloseencounters}. The occurrences, position and length of each colored line provide information on the number, moment and duration in proper time of the corresponding encounters along every orbit, where by `duration' we mean the time spent in the Hill's sphere.\\
\indent The last result that we report is Fig. \ref{fig:Tiss}.
\begin{figure}
	\centering
	\includegraphics[scale=0.52]{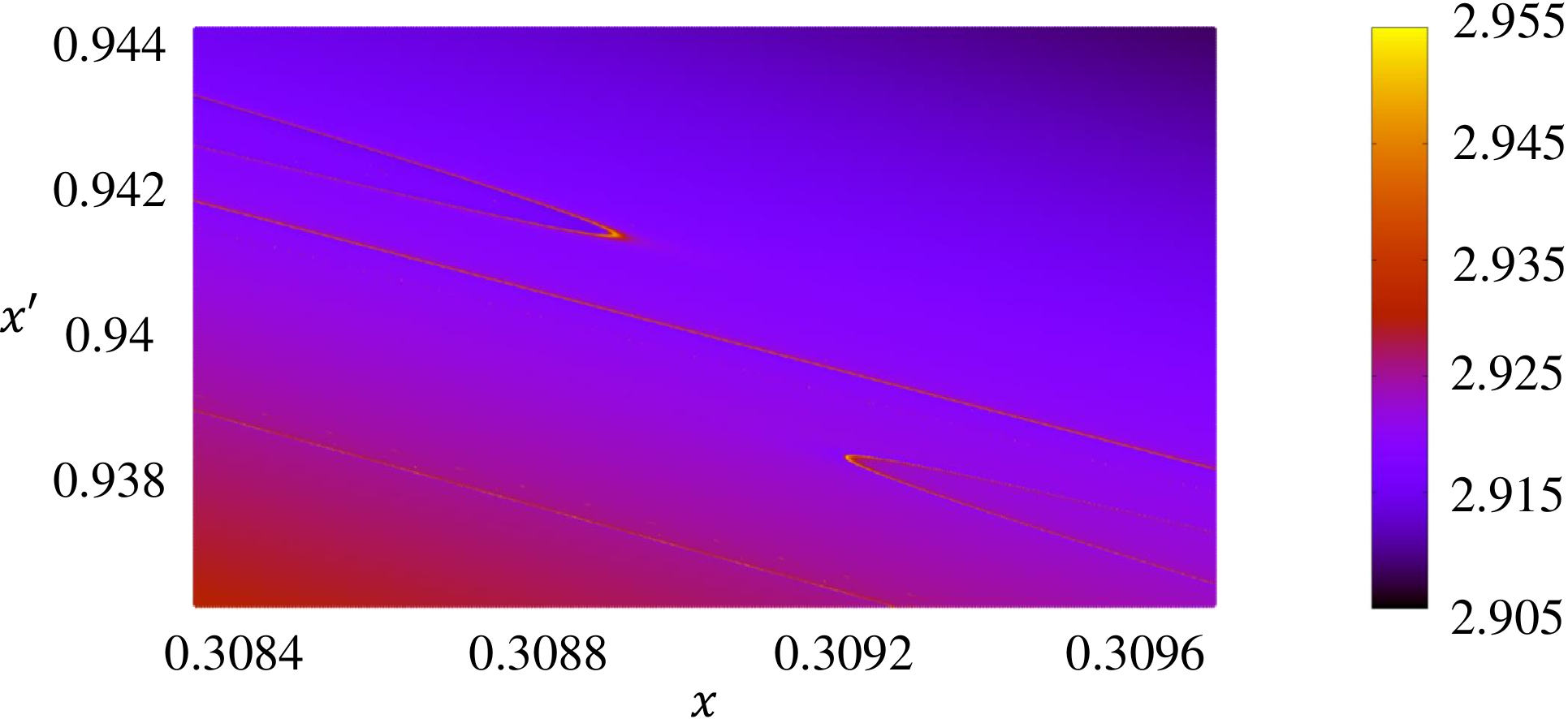}
	\caption{$\mathcal{T}(F)$ from \eqref{eqn:Tiss} over the mFLI central magnification in the planar case.}
	\label{fig:Tiss}
\end{figure}
On the same grid $(x,x')$ of the central zoom-in of top panel of Fig. \ref{fig:mFLIEarth}, we compute the Tisserand parameter 
\begin{equation}
\label{eqn:Tiss}
	\mathcal{T}=\frac{1}{a_P}+2\cos i_P\sqrt{a_P(1-e_P^2)}
\end{equation}   
at final value $f=F$, which is subject to a more significant variation when close approaches occur. Comparing to Fig. \ref{fig:mFLIEarth}, we can notice that only fewer and less detailed structures are replicated, like the three main lobes identified above, corresponding to the deepest encounters. This further supports the effectiveness of the indicator \eqref{modflireg}.\\

\vskip 0.5 cm

\noindent\textbf{Acknowledgements.} {The author M.R. acknowledges the project ``AIxtreme: Physics-driven AI approaches for predicting extreme weather and space weather events'' funded by Compagnia di San Paolo Foundation. The author M.G. acknowledges the project MIUR-PRIN 2020RC5H82 titled ``Modern challenges of Celestial Mechanics: from the fundamental theorems to the new models of Planetary Sciences and back''.}

\appendix

\section{Appendix}
\label{appendix:chcoord}
\subsection{From the rotating-pulsating to the inertial reference frame}
\label{subsecapp:xyztoXYZ}
{
In the ER3BP the primary and the secondary bodies rotate non-uniformly around their common center of mass and have constantly varying relative distance. The rotating-pulsating frame $Oxyz$ then represents the straightforward extension of the synodic frame of the circular case by:
\begin{enumerate}[label=(\roman*)]
	\item 
	rotating by $f(t)$, thus with non-constant angular speed $\dot{f}$;
	\item 
	pulsating in order to rescale lengths by a factor $1/\varrho(f(t))$, where
	\begin{equation}
		\label{eqn:2BpolER3BP}
		\varrho(f)=\frac{1-\varepsilon^2}{1+\varepsilon\cos f}
	\end{equation}
	is the primaries' Keplerian ellipse according to units in Section \ref{sec:intro}.
\end{enumerate} 
In this way $P_1$ and $P_2$ appear at rest and $\Norm{P_2-P_1}=1$.\\
The two prescriptions translate mathematically in the application of a rotation matrix $\mathscr{R}(f)\in SO(3)$ and the scaling factor $\varrho(f)$ to a vector $r=(x,y,z)\in\mathbb{R}^3$ in $Oxyz$ to retrieve a vector $R=(X,Y,Z)\in\mathbb{R}^3$ in the inertial frame $OXYZ$:
\begin{equation}
	\label{eqn:ER3BPRP}
	R=\varrho(f)\mathscr{R}(f)r\;,
\end{equation}
where
\begin{equation}
	\label{eqn:ER3BPRPmatr}
	\mathscr{R}(f)=
	\begin{pmatrix}
		\cos f & -\sin f & 0\\
		\sin f & \cos f & 0\\
		0 & 0 & 1
	\end{pmatrix}\;.
\end{equation}
}

\subsection{Implementation procedure}
\label{subsecapp:implem}
We rely to the following roadmap to conduct the regularization numerically:
\begin{enumerate}[label=(\roman*)]
\item \textbf{Initial conditions}. Given an initial Cartesian datum
  $(x_0,y_0,z_0,p_{x,0},p_{y,0},p_{z,0})$ at $f=f_0$, define:
	$$
	\Phi_0=-\mathcal{H}(x_0,y_0,z_0,p_{x,0},p_{y,0},p_{z,0},f_0)\;,\quad \phi_0=f_0\;.
	$$
	\item \textbf{Transformed initial conditions}.
	The correspondence between
	\begin{equation*}
	(u_1(0),u_2(0),u_3(0),u_4(0),\phi(0),U_1(0),U_2(0),U_3(0),U_4(0),\Phi(0))
	\end{equation*} 
	and
	\begin{equation*}
	(x(f_0),y(f_0),z(f_0),\phi(f_0),p_1(f_0),p_2(f_0),p_3(f_0),\Phi(f_0))
	\end{equation*}
	is clearly not one-to-one. In order to close the problem, it is sufficient to consider the following local pre-images of the map $\pi(u)$  \cite{Froeschle1970} extended to the momenta and that identically fulfill the condition on the bilinear form:
	\begin{equation}
		\label{eqn:invpibar}
		\begin{array}{c}
			\bar{\pi}_-^{-1}\colon T^*((\mathbb{R}^3\setminus\{(q_1,0,0)\colon q_1\ge 0\})\times\mathbb{T})\longrightarrow T^*((\mathbb{R}^4\setminus\mathscr{C})\times\mathbb{T})\cap \{l = 0\}
			\\
			\bar{\pi}_+^{-1}\colon T^*((\mathbb{R}^3\setminus\{(q_1,0,0)\colon q_1\le 0\})\times\mathbb{T})\longrightarrow T^*((\mathbb{R}^4\setminus\mathscr{C})\times\mathbb{T})\cap \{l= 0\}
		\end{array}
	\end{equation}
	such that:
	\begin{align}
		\label{eqn:invpibarrel}
		\begin{split}
		  \bar{\pi}_-^{-1}(q,\phi,\bar{p},\Phi)=\left(\pi_-^{-1}(q),\phi,2A(\pi_-^{-1}(q))^T
                  (\bar{p},0),\Phi\right)\\
			\bar{\pi}_+^{-1}(q,\phi,\bar{p},\Phi)=\left(\pi_+^{-1}(q),\phi,2A(\pi_+^{-1}(q))^T(\bar{p},0),\Phi\right)\\
		\end{split}\;,
	\end{align}
	where
	\begin{equation}
		\label{eqn:pbar}
	\bar{p}=(p_1,p_2-1+\mu,p_3)\;,
	\end{equation}
	and
	\begin{align}
		\label{eqn:invpi}
		\begin{split}
			\pi_-^{-1}(q)=\left(\frac{q_2}{\sqrt{2(d-q_1)}},\sqrt{\frac{d-q_1}{2}},0,\frac{q_3}{\sqrt{2(d-q_1)}}\right)\\
			\pi_+^{-1}(q)=\left(\sqrt{\frac{d+q_1}{2}},\frac{q_2}{\sqrt{2(d+q_1)}},\frac{q_3}{\sqrt{2(d+q_1)}},0\right)
		\end{split}
	\end{align}
	with $d=d_2=\sqrt{q_1^2+q_2^2+q_3^2}$.

	\item \textbf{Original solutions}. Numerically integrate the 
          Hamilton equations of the regularized Hamiltonian  $\mathcal{K}$ and retrieve the
          solutions of the Hamilton equations of the original Hamiltonian $\mathcal{H}$ with a simple
          projection:
	\begin{equation}
		\label{eqn:KSsyst}
		\begin{dcases}
			x_k=u_{1,k}^2-u_{2,k}^2-u_{3,k}^2+u_{4,k}^2+1-\mu\\
			y_k=2u_{1,k}u_{2,k}-2u_{3,k}u_{4,k}\\
			z_k=2u_{3,k}u_{1,k}+2u_{4,k}u_{2,k}\\
			p_{x,k}=\frac{1}{2\Vert u_k\Vert^2}(u_{1,k}U_{1,k}-u_{2,k}U_{2,k}-u_{3,k}U_{3,k}+u_{4,k}U_{4,k})\\
			p_{y,k}=\frac{1}{2\Vert u_k\Vert^2}(u_{2,k}U_{1,k}+u_{1,k}U_{2,k}-u_{4,k}U_{3,k}-u_{3,k}U_{4,k})+1-\mu\\
			p_{z,k}=\frac{1}{2\Vert u_k\Vert^2}(u_{3,k}U_{1,k}+u_{4,k}U_{2,k}+u_{1,k}U_{3,k}+u_{2,k}U_{4,k})\\
		\end{dcases}\;
	\end{equation}
	at the  sample points $f_k=f(s_k)$, $k=1,2,\ldots,K$, given by numerically
        integrating \eqref{eqn:fictrueanfofs}.
\end{enumerate}


\begin{thebibliography}{70}
\bibitem{aarseth-zare}
S.J. Aarseth and K. Zare. ``A regularization of the three-body problem''. In: \emph{CM\&DA}, 10 (1974), pp. 185-205.
	
\bibitem{aarseth2003}
S.J. Aarseth. ``Gravitational N-Body Simulations, Tools and Algorithms''. Cambridge Univ. Press, New York, 2003.
 
\bibitem{Arenstorf:regER3BP}
R.F. Arenstorf. ``Regularization theory for the elliptic restricted three body problem''. In: \emph{JDE}, 6.3 (1969), pp. 420-451.

\bibitem{Birkhoff:restricted}
G.D. Birkhoff. ``The restricted problem of three bodies''. In: \emph{Rendiconti del Circolo Matematico di Palermo (1884-1940)}, 39.1 (1915), p. 265.

\bibitem{breiter2017}
S. Breiter and K. Langner. ``Kustaanheimo-Stiefel transformation with an arbitrary defining vector''. In: \emph{CM\&DA}, 128 (2017), pp. 323–342. 

\bibitem{Broucke}
R. Broucke. ``Periodic collision orbits in the elliptic restricted three-body problem''. In: \emph{CM\&DA}, 3 (1971), pp. 461-477.

\bibitem{GuzzoCardin:intCR3BP}
F. Cardin and M. Guzzo. ``Integrability of close encounters in the spatial restricted three-body problem''. In: \emph{Comm. Cont. Math.}, (2022), 24.06:2150040.

\bibitem{CLSF}
A. Celletti, E. Lega, L. Stefanelli and C. Froeschl\'e. ``Some results on the global dynamics of the regularized restricted three-body problem with dissipation''. In: \emph{CM\&DA}, 109 (2011), pp. 265–284.

\bibitem{Falcolini}
C. Falcolini. ``Perturbative Methods in Regularization Theory''. In: Benest, D., Froeschlé, C., (eds) \emph{Singularities in Gravitational Systems}. Lecture Notes in Physics, Vol. 590 (2002). Springer, Berlin, Heidelberg.

\bibitem{FNS}
J. Font, A. Nunes and C. Sim\'o. ``Consecutive quasi-collisions in the planar circular RTBP''. In: \emph{Nonlinearity}, 15 (2002), p. 115.

\bibitem{Froeschle1970}
C. Froeschl\'e. ``Numerical study of dynamical systems with three degrees of freedom I. Graphical displays of Four-dimensional sections''. In: \emph{A\&A}, 4 (1970), pp. 115-128. 

\bibitem{GKZ}
M. Guardia, V. Kaloshin and J. Zhang. ``Asymptotic Density of Collision Orbits in the Restricted Circular Planar 3 Body Problem''. In: \emph{Archive for Rational Mechanics and Analysis}, Vol. 233, Issue 2 (2019), pp 799-836.

\bibitem{GL2013}
M. Guzzo and E. Lega. ``On the identification of multiple close-encounters in the planar circular restricted three body problem''. In: \emph{MNRAS}, 428 (2013), pp. 2688-2694. 

\bibitem{GL14}
M. Guzzo M. and E. Lega. ``Evolution of the tangent vectors and localization of the stable and unstable manifolds of hyperbolic orbits by Fast Lyapunov Indicators''. In: \emph{SIAM J. Appl. Math.}, Vol. 74 (2014), No. 4, pp. 1058-1086.

\bibitem{GL15}
M. Guzzo and E. Lega. ``A study of the past dynamics of comet 67P/Churyumov-Gerasimenko with Fast Lyapunov Indicators''. In: \emph{A\&A}, 579(A79) (2015), pp. 1-7.

\bibitem{GL17}
M. Guzzo and E. Lega. ``Scenarios for the dynamics of comet 67P/Churyumov-Gerasimenko over the past 500 kyr''. In: \emph{MNRAS}, 469 (2017), pp. S321-S328.

\bibitem{GL18}
M. Guzzo and E. Lega. ``Geometric chaos indicators and computations of the spherical hypertube manifolds of the spatial circular restricted three-body problem''. In: \emph{Physica D}, 373 (2018), pp. 38–58.

\bibitem{Guzzo}
M. Guzzo. ``Parametric approximations of fast close encounters of the planar three-body problem as arcs of a focus-focus dynamics''. ArXiv preprint arXiv:2311.17108 (2023).

\bibitem{GL23}
M. Guzzo and E. Lega. ``Theory and applications of fast Lyapunov indicators to model problems of celestial mechanics''. In: \emph{CM\&DA}, (2023), 135.4:37.

\bibitem{heggie}
C. Heggie. ``A global regularization of the gravitational N-body problem''. In: \emph{CM\&DA}, 10 (1974), pp. 217-241. 

\bibitem{henrard}
J. Henrard. ``On Poincar\'e's second species solutions''. In \emph{CM\&DA}, 21 (1980), pp. 83-97.

\bibitem{K64}
P. Kustaanheimo. ``Spinor regularisation of the Kepler motion''. In: \emph{Annales Universitatis Turkuensis A 73}, 73 (1964), pp. 1-7.

\bibitem{KS65}
P. Kustaanheimo and E.L. Stiefel. ``Perturbation theory of Kepler motion based on spinor regularization''. In: \emph{J. fur die Reine und Angewandte Mathematik}, 218 (1965), pp. 204-219-569.

\bibitem{langnerbreiter2015}
K. Langner and S. Breiter. ``KS variables in rotating reference frame. Application to cometary dynamics''. In: \emph{Astrophysics and Space Science}, (2015), 357:153.

\bibitem{LGF11} 
E. Lega, M. Guzzo and C. Froeschl\'e. ``Detection of close encounters and resonances in three-body problems through Levi-Civita regularization''. In: \emph{MNRAS}, 418 (2011), pp. 107-113.

\bibitem{LG16}
E. Lega and M. Guzzo. ``Three-dimensional representations of the tube manifolds of the planar restricted three-body problem''. In: \emph{Physica D}, 325 (2016), pp. 41-52.

\bibitem{LC1906} 
T. Levi-Civita. ``Sur la régularisation qualitative du probléme restreint des trois corps''. In: \emph{Acta Math.}, 30 (1906), pp. 305-327.

\bibitem{Llibrepinol}
J. Llibre and C. Pi\~nol. ``On the elliptic restricted three-body problem''. In: \emph{CM\&DA}, 48 (1990), pp. 319-345.  

\bibitem{luther1968explicit}
H. A. Luther. ``An explicit sixth-order Runge-Kutta formula''. In: \emph{Mathematics of Computation}, 22.102 (1968), pp. 434-436.

\bibitem{marsden1974reduction}
J. Marsden and A. Weinstein. ``Reduction of symplectic manifolds with symmetry''. In: \emph{Reports on mathematical physics}, 5.1 (1974), pp. 121-130.

\bibitem{meyer1973symmetries}
K.R. Meyer. ``Symmetries and integrals in mechanics''. In: \emph{Dynamical systems. Academic Press}, (1973), pp. 259-272.

\bibitem{opik}
E.J. \"Opik. ``Interplanetary Close Encounters''. Elsevier, New York (1976).

\bibitem{PG21}
R.I. Paez and M. Guzzo. ``Transits close to the Lagrangian solutions $L_1$ , $L_2$ in the elliptic restricted three-body problem''. In: \emph{Nonlinearity}, 34 (2021), pp. 6417-6449.

\bibitem{Pinyol}
C. Pinyol. ``Ejection-collision orbits with the more massive primary in the planar elliptic restricted three-body problem''. In: \emph{CM\&DA}, 61 (1995), pp. 315-331.

\bibitem{RECharac}
M. Rossi and C. Efthymiopoulos. ``Characterization of the stability for trajectories exterior to Jupiter in the restricted three-body problem via closed-form perturbation theory''. In: \emph{Proceedings of the International Astronomical Union}, 15.S364 (2019), pp. 232-238.

\bibitem{REReleg}
M. Rossi and C. Efthymiopoulos. ``Relegation-free closed-form perturbation theory and the domain of secular motions in the restricted three-body problem''. In: \emph{CM\&DA}, 135.4 (2023), p. 42.

\bibitem{saha}
P. Saha. ``Interpreting the Kustaanheimo-Stiefel Transform in Gravitational Dynamics''. In: \emph{MNRAS}, 400-1 (2009), pp. 228-231.

\bibitem{scheibner1866satz}
W. Scheibner. ``Satz aus der Störungstheorie.(Auszug aus einem Schreiben an den Herausgeber)''. In: \emph{Journal für die reine und angewandte Mathematik}, 1866.65 (1866), pp. 291-292.

\bibitem{shefer1990}
V.A. Shefer. ``Application of KS-transformation in the problem of investigation of the motion of unusual minor planets and comets''.  In: \emph{CM\&DA}, 49 (1990), pp. 197-207. 

\bibitem{StiefelScheile1971}
E.L. Stiefel and G. Scheifele. ``Linear and Regular Celestial Mechanics''. Grundlehren der Mathematischen Wissenschaften. Springer, Berlin (1971).

\bibitem{Szebehely:regER3BP}
V. Szebehely and G.E.O. Giacaglia. ``On the elliptic restricted problem of three bodies''. In: \emph{The Astronomical Journal}, 69 (1964), p. 230.

\bibitem{SZ:thorb}
V. Szebehely. ``Theory of orbits: the restricted problem of three bodies''. Tech. rep. Yale univ. New Haven CT, 1967.

\bibitem{valsecchi2002}
G.B. Valsecchi. ``Close Encounters in \"Opik Theory''. In: Benest, D., Froeschlé, C., (eds) \emph{Singularities in Gravitational Systems}. Lecture Notes in Physics. Springer (2002).

\bibitem{Waldvogel:regER3BP}
J. Waldvogel. ``Die Verallgemeinerung der Birkhoff-regularisierung für das räumliche Dreikörperproblem''. PhD thesis. ETH Zurich, 1966.

\bibitem{Waldvogel:NASA}
J. Waldvogel. ``The restricted elliptic three-body problem''. In: Stiefel, E., R\"ossler, M., Waldvogel, J., Burdet, C.A., (eds.) \emph{Methods of Regularization for Computing Orbits in Celestial Mechanics}. NASA Contractor Report NASA CR 769 (1967), pp. 88–115. 

\bibitem{Waldvogel:2006}
J. Waldvogel. ``Quaternions and the perturbed Kepler problem''. In: \emph{CM\&DA}, 95 (2006), pp. 201-212.

\bibitem{ZhaoRCD}
L. Zhao. ``Kustaanheimo-Stiefel Regularization and the Quadrupolar Conjugacy''. In: \emph{Regular and Chaotic Dynamics}, 20-1 (2015), pp. 19-36.
\end{thebibliography}
\end{document}